\documentclass[aps,twocolumn]{revtex4-1}
\usepackage{graphicx}
\usepackage{amsmath}
\usepackage{amssymb}
\usepackage{amsthm}
\usepackage{color}
\usepackage{hyperref}

\theoremstyle{plain}
\newtheorem{thm}{Theorem}

\newtheorem{cor}[thm]{Corollary}
\newtheorem{lem}[thm]{Lemma}
\newtheorem{prop}[thm]{Proposition}

\theoremstyle{definition}
\newtheorem{defi}{Definition}

\newtheorem*{conj}{Conjecture}

\theoremstyle{remark}
\newtheorem*{rem}{\bf Remark}

\definecolor{labelkey}{rgb}{0,.56,.7}

\newcommand{\haf}{\mathop{{\mathrm{haf}}}\nolimits}

\DeclareMathOperator*{\spec}{spec}
\DeclareMathOperator*{\diag}{diag}

\def\bbZ{\mathbb{Z}}
\def\bbR{\mathbb{R}}

\def\bbI{\mathbb{I}}
\def\bbJ{\mathbb{J}}

\DeclareSymbolFont{Eulerscripteusm10}{U}{eus}{m}{n}
\SetSymbolFont{Eulerscripteusm10}{bold}{U}{eus}{b}{n}
\DeclareMathSymbol{\euI}{\mathord}{Eulerscripteusm10}{"4A}
\DeclareMathSymbol{\euK}{\mathord}{Eulerscripteusm10}{"4B}
\DeclareMathSymbol{\euR}{\mathord}{Eulerscripteusm10}{"52}

\DeclareMathAlphabet{\pazocal}{OMS}{zplm}{m}{n}   

\def\a{\alpha}
\def\b{\beta}

\def\D{\Delta}

\def\om{\omega}

\def\s{\sigma}

\def\la{\lambda}

\def\t{\theta}
\def\om{\omega}

\begin{document}

\title{Gaussian Boson Sampling for perfect matchings of arbitrary graphs}

\author{Kamil Br\'adler}
\email{kamil@xanadu.ai}
\author{Pierre-Luc Dallaire-Demers}
\author{Patrick Rebentrost}
\author{Daiqin Su}
\email{daiqin@xanadu.ai}
\author{Christian Weedbrook}

\affiliation{Xanadu, 372 Richmond Street West, Toronto, Ontario M5V 1X6, Canada }


\date{\today}

\begin{abstract}
{A famously hard graph problem with a broad range of applications is computing the number of  perfect matchings, that is the number of unique and complete pairings of the vertices of a graph.  We propose a method to estimate the number of perfect matchings of undirected graphs based on the relation between Gaussian Boson Sampling and graph theory. The probability of measuring zero or one photons in each output mode is directly related to the hafnian of the adjacency matrix, and thus to the number of perfect matchings of a graph. We present encodings of the adjacency matrix of a graph into a Gaussian state and show strategies to boost the sampling success probability. With our method, a Gaussian Boson Sampling device can be used to estimate the number of perfect matchings significantly faster and with lower energy consumption compared to a classical computer. 
}
\end{abstract}

\maketitle


\section{Introduction}

Graphs appear in a variety of situations, ranging from computer science and the physical sciences to economics and social sciences  \cite{Barabasi1999,clrs2001algorithms,Otte2002}. For example, graphs model the flow of computation in an algorithm, the structure of a molecule or the molecular graph of a photosynthetic light-harvesting organism \cite{Rebentrost2009}, the supply chain of an online retailer, or the network of people on social media. Graphs consist of vertices representing people, atoms, or warehouses, and edges representing connections between the vertices. Solving problems on graphs can help identifying highly connected influencers in a social network, improving the efficiency of solar cells \cite{Wodo2012}, or reducing the time till a customer receives an online order.

Common computational problems on graphs involve finding subgraphs with defined properties such as high connectivity (cliques), routing problems such as the traveling salesman problem, network flow problems, or the coloring of graphs \cite{west2001introduction}.
Many graph problems are easy to compute, that is they run in time that is at most a polynomial in the number of vertices and edges. For example, the shortest path in a graph or the maximum matching in a graph can be solved in polynomial time.
Other problems are famously hard to solve \cite{karp1972combinatorial}, running in a time that is exponential in either the number of vertices or the number of edges. Examples include the NP-complete clique  and traveling salesman problems.

Quantum algorithms promise to solve certain problems much faster than classical algorithms \cite{Montanaro2016}. Recently, it was discovered that the output probability distribution of photons from a multiport whose input is a Gaussian squeezed state~\cite{hamilton2016gaussian} depends on the~\emph{hafnian} of a  sampling matrix. The method, called Gaussian Boson Sampling, is a natural generalization of the original photonic boson sampling~\cite{aaronson2011computational}, and its modifications~\cite{lund2014boson,huh2017vibronic,huh2017vibronic,spring2012boson,olson2015sampling}, where single photons or vacuum are input states of a linear interferometer. The hafnian was introduced in the context of interacting quantum field theory~\cite{caianiello1953quantum} and belongs to the class of matrix functions together with the permanent, Pfaffian or determinant~\cite{minc1984permanents}. Unlike the latter two functions, the permanent and the hafnian are believed to be hard to calculate since they belong  to the $\#$P-complete complexity class and both functions (the permanent in particular) have been studied in detail~\cite{glynn2010permanent,valiant1979complexity,troyansky1996quantum,jerrum2004polynomial,rudelson2016hafnians,bjorklund2012counting,barvinok2017combinatorics,scheel2004permanents}. The class $\#$P involves counting problems such as how many solutions exist to a
given problem in the class~NP.

In this work, we connect Gaussian Boson Sampling to graph theory in order to tackle graph problems that involve the number of perfect matchings~\cite{lawler1976combinatorial,wolsey1999combinatorial}. This connection is based on the fact that  the number of perfect matchings of an undirected graph
equals to the hafnian of the graph's adjacency matrix. We propose an efficient way to embed an arbitrary adjacency matrix into a valid Gaussian state.
The desired Gaussian state can be generated by injecting thermal states or vacuum states into a circuit consisting of single-mode squeezers and a multiport. The information of the
required input states, single-mode squeezers and the multiport is completely determined by the adjacency matrix. Thus, the whole procedure is programable
and is straightforward to implement, and therefore can allow the estimation of the number of perfect matchings in arbitrary graphs. 
Promisingly, by appropriate rescaling of the problem and  linearly increasing the size of the circuit we can substantially improve the sampling success probability. 
We demonstrate our methods for the special cases of complete graphs and one-edge-removed complete graphs.

This paper is organized as follows. Sec. \ref{sec:results} summarizes the main results. 
We first discuss embedding an adjacency matrix into a mixed or pure Gaussian state, and connect 
the hafnian of the adjacency matrix to the measurement probability. We then discuss the hardness of estimating the hafnian using Gaussian Boson Sampling and propose 
strategies to amplify the measurement probability. We finally discuss a complete graph example to illustrate the validity of the method. We conclude in Sec. \ref{sec:conclusion}. 

\section{Results}\label{sec:results}
\subsection{Gaussian Boson Sampling and hafnian}
As shown in \cite{hamilton2016gaussian} the output probability distribution of photons from the multiport with a Gaussian input state is related to the hafnian. In particular, the authors showed that the probability of the photon number distribution, with no more than one photon in each output mode, is related to the hafnian of the submatrices of a matrix $A$, derived from the Gaussian covariance matrices $\s_A$. For an arbitrary $M$-mode Gaussian state,
$A$ is related to $\s_A$ via \cite{hamilton2016gaussian}
\begin{equation}\label{eq:sigmaA-1}
A=X_{2M} \big[ \bbI_{2M}- (\sigma_A+\bbI_{2M}/2)^{-1} \big],
\end{equation}
where 
$\bbI_{2M}$ is the identity matrix and $X_{2M}$ is defined as
$
X_{2M}=\begin{bmatrix}
           0 & \bbI_M \\
           \bbI_M & 0 \\
         \end{bmatrix}.
$
In more detail, assume that $n = \bigotimes_j^M n_j | n_j\rangle \langle n_j|$ represents a certain measurement pattern of photons, where $n_j$ is the photon number in the $j$-th mode. In the case $n_j = \{ 0, 1\}$, the measurement probability is given by~\cite{hamilton2016gaussian}
\begin{equation}\label{eq:ProHafOrig}
  \mathrm{Pr} (\bar n) = {1\over \bar n! \sqrt{\det{\sigma_Q}}} \haf{A_S},
\end{equation}
where $\bar n = n_1! n_2! \dots n_M!$, $\s_Q \equiv \s_A + \bbI_{2M}/2$ and $A_S$ is the submatrix of $A$ determined by the position of the output modes registering single photons.
In particular, finding the probability of detecting a single photon in every output mode corresponds to estimating the hafnian of the full matrix $A$.

\begin{figure}[]
\includegraphics[width=8.6cm]{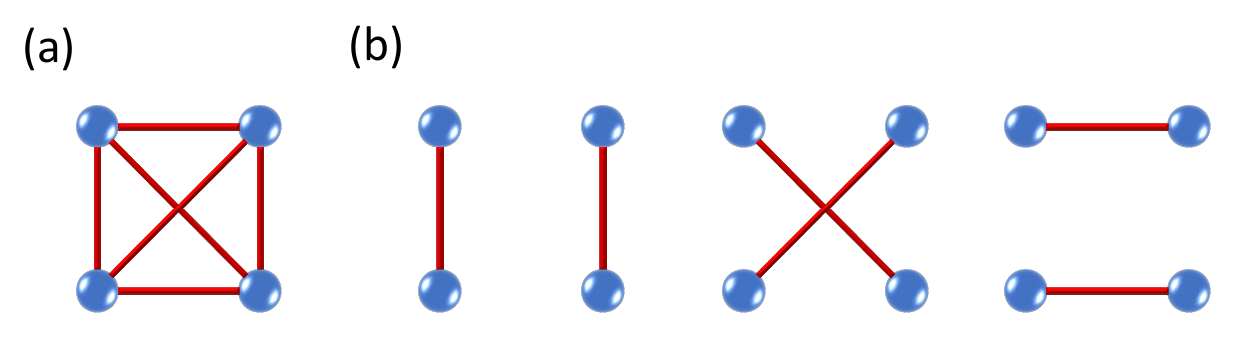}
\caption{ a) Complete graph of four vertices. b) Three perfect matchings of the four vertex complete graph. }
\label{fig:PM}
\end{figure}

A perfect matching of a graph is when every vertex of the graph is matched up with a single unique other vertex. Often there are multiple ways of perfectly matching vertices, as demonstrated for a small example in Fig.~\ref{fig:PM}.  Estimating the number of such perfect matchings is known to be difficult for a classical computer as the size of the graph increases \cite{bjorklund2012counting, Cygan2015}.
It is also known that the number of perfect matchings is equal to the hafnian of a graph's adjacency matrix. Motivated from Eq. \eqref{eq:ProHafOrig},
we ask whether it is possible to estimate the number of perfect matchings via measuring the probability of the photon number distribution.

The key result of this work is to map general graphs to the optical setup of Gaussian Boson Sampling. We start first with inverting Eq.~(\ref{eq:sigmaA-1}) as
\begin{equation}\label{eq:sigmaA}
\sigma_A=(\bbI_{2M}-X_{2M}A)^{-1}-{\bbI_{2M}\over2}.
\end{equation}
Given this relation, a naive idea is to consider $A$ as the adjacency matrix of a graph and
insert it into Eq. (\ref{eq:sigmaA}) to find the corresponding covariance matrix. Unfortunately, this fails because a general adjacency matrix
does not map to a valid covariance matrix through Eq. (\ref{eq:sigmaA}). An adjacency matrix is a real symmetric matrix and has entries either zero or one. The hafnian
of a large size adjacency matrix is usually a large number, e.g., a complete graph with $2M$ vertices has $(2M-1)!!$ perfect matchings. 
Intuitively, Eq. (\ref{eq:ProHafOrig}) does not hold if $A$ is simply regarded as an adjacency matrix
because the probability is a number smaller than one. To make things reasonable one has to rescale the adjacency matrix by a small number.

We develop a systematic method to map an arbitrary adjacency matrix of an undirected graph to a valid covariance matrix of a Gaussian state (see Appendix \ref{sec:GraphToCM} for details). We associate the number of perfect
matchings with the probability of a photon number distribution. A covariance matrix is by its nature a Hermitian and positive-definite matrix, and satisfies the uncertainty principle \cite{simon1994quantum}. Following these three constraints, we vary an arbitrary adjacency matrix such that it maps to a valid covariance matrix.

\subsection{Sampling mixed Gaussian covariance matrices}\label{sec:SampleMixState}

\begin{figure*}[ht!]
\includegraphics[width=15cm]{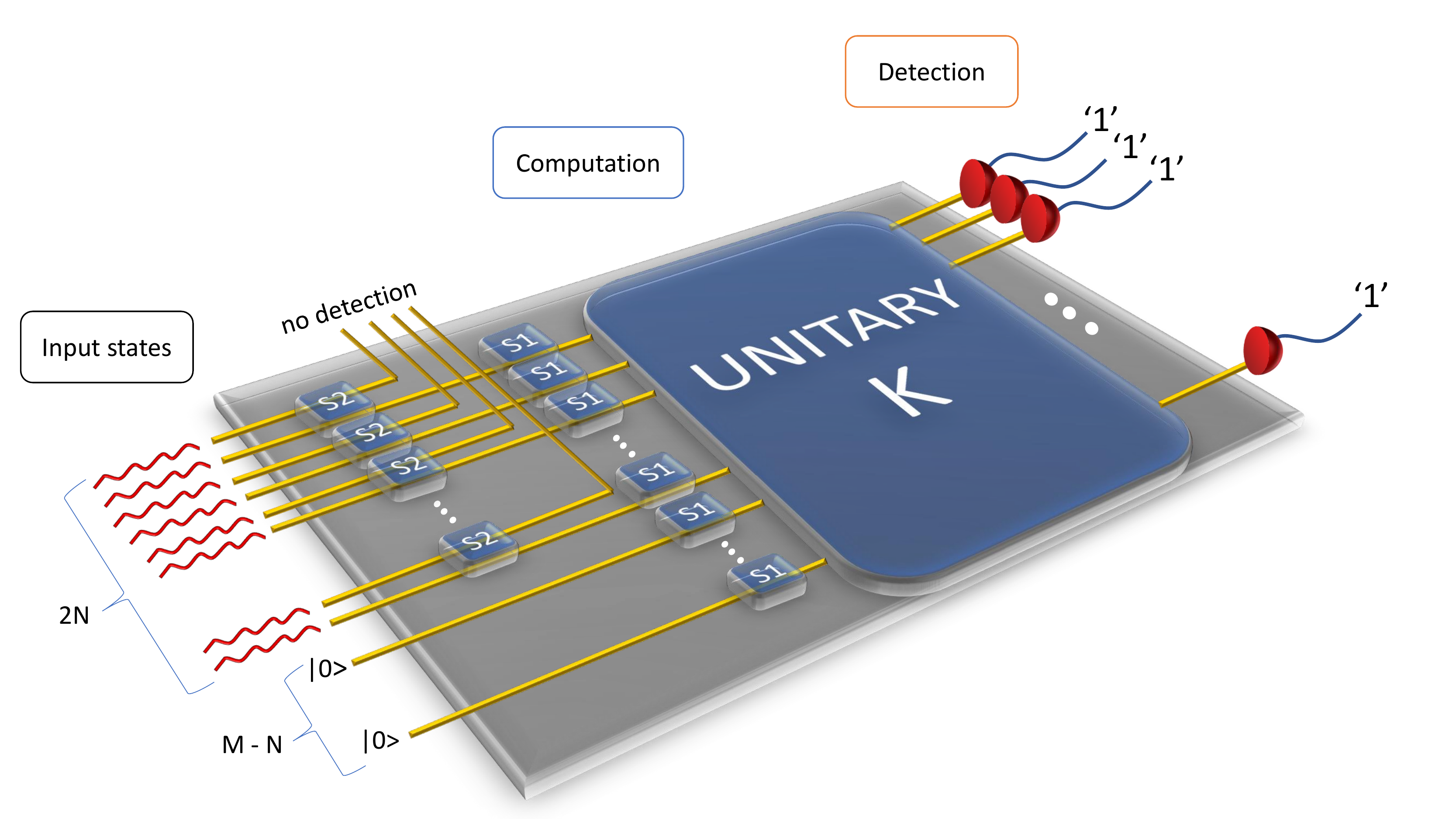}
\caption{\small $(M+N)$-mode circuit that generates $M$-mode mixed Gaussian covariance matrices $\s_{cA}$. $N$ input thermal states are generated by
tracing out the other $N$ modes of $N$ two-mode squeezed vacuum states with squeezing parameters $\xi_k$.
Then $N$ single-mode squeezers squeeze the thermal states and the other
$(M-N)$ single-mode squeezers with squeezing parameters $r_k$ squeeze the vacua. Photon number detectors measure the state after the linear optics network $K$. }
\label{fig:MixCircuit}
\end{figure*}

Here, we discuss the conditions under which we can map a graph adjacency matrix $A$ to a mixed Gaussian covariance matrix.
A $2M\times 2M$ adjacency matrix $A$ can be written in a block form, with four $M \times M$ submatrices $A_{11}$, $A_{12}$, $A_{21}$ and $A_{22}$. If these submatrices satisfy
the following conditions: (1) $A_{11} = A_{22}$ and $A_{12} = A_{21}$, (2) $A_{11}$ commutes with $A_{12}$, (3) $A_{12}$ is positive definite; then $c A$ maps to a valid
covariance matrix $\s_{cA}$ via Eq. (\ref{eq:sigmaA}), where $c$ is a rescaling parameter satisfying $0 < c < 1/\lambda_1$ and $\lambda_1$ is the maximal eigenvalue of $A$.
Generally, $\s_{cA}$ is an $M$-mode mixed Gaussian covariance matrix. 
This can be seen from its symplectic eigenvalues \cite{weedbrook2012gaussian} (see also Appendix \ref{app:PPE}), given by
\begin{equation}\label{SymplecticEV}
\nu_k = \frac{1}{2} \sqrt{\frac{(1+c h_k)^2 - c^2 f_k^2}{(1-c h_k)^2 - c^2 f_k^2}} \ge \frac{1}{2},\quad k = 1, 2, \hdots, M,
\end{equation}
where $f_k$ and $h_k$ are eigenvalues of $A_{11}$ and $A_{12}$, respectively. Note that $f_k \pm h_k$ are eigenvalues of the adjacency matrix $A$. 
A Gaussian state is pure if $\nu_k = 1/2$ for all $k$ \cite{weedbrook2012gaussian}. From Eq. (\ref{SymplecticEV}), this happens 
when all eigenvalues of $A_{12}$ are zero, namely, $A_{12}=0$.

Preparing a general mixed Gaussian state is not usually an easy task. According to Williamson's theorem~\cite{williamson1936algebraic},
a covariance matrix can be put into a diagonal form by a symplectic transformation.
The symplectic transformation can be further decomposed into a product of an orthogonal matrix $K$, a direct sum of $M$ single-mode squeezing matrices
and another orthogonal matrix $L$. This implies a mixed Gaussian state can be generated by injecting single-mode thermal states (or vacuum) into an $M$-mode
linear optics network $L$, passing through $M$ single-mode squeezers and then another linear optics network $K$.

We observe that to produce the covariance matrix $\s_{cA}$, the orthogonal matrix $L$ can be set to an identity. Therefore the state before entering the linear
optics network $K$ is a product of squeezed thermal states or squeezed vacuum states. The thermal state is completely characterized
by the symplectic eigenvalues $\nu_k$, Eq. \eqref{SymplecticEV}, which is dependent on the property of the adjacency matrix $A$ and the rescaling number $c$.
The squeezing parameters $r_k$ of $M$ single-mode squeezers are also determined by the adjacency matrix and $c$ 
from $e^{4r_k} = \frac{(1+c |f_k|)^2-c^2 h_k^2}{(1-c |f_k|)^2-c^2 h_k^2}$.
A single-mode thermal state can be purified with another thermal state, forming a pure two-mode squeezed state. The two-mode squeezing parameter $\xi_k$ is related to the symplectic eigenvalue $\nu_k$ via $2\nu_k = \cosh{[2\xi_k]}$. Suppose that $A_{12}$ has $N$ nonzero eigenvalues,
then $N$ of the $M$ input modes should be injected by thermal states and others by vacuum states. To purify the $N$ thermal states we introduce $N$ two-mode squeezers which squeeze the vacuum. The $(M+N)$-mode circuit that can produce a $M$-mode mixed Gaussian state is shown in Fig. \ref{fig:MixCircuit}.

If we already prepared the state $\s_{cA}$ via the circuit in Fig. \ref{fig:MixCircuit}, we then detect all output modes and evaluate the probability of detecting a single
photon in every output mode in order to sample the full adjacency matrix. The probability is related to the hafnian of the adjacency matrix $A$ by
\begin{align}\label{eq:probA}
&\mathrm{Pr}_{cA}(\underbrace{1,\dots,1}_{M})  \nonumber\\
&= \prod_{k=1}^{M} \sqrt{[1-c(f_k + h_k)] [ 1+c(f_k - h_k)]} ~c^M \haf{A}. \nonumber\\
\end{align}
Therefore we can obtain the hafnian of $A$ once we know the probability.

\subsection{Sampling pure Gaussian covariance matrices}\label{sec:SamplePureState}

\begin{figure*}[ht!]
  \resizebox{12cm}{!}{\includegraphics{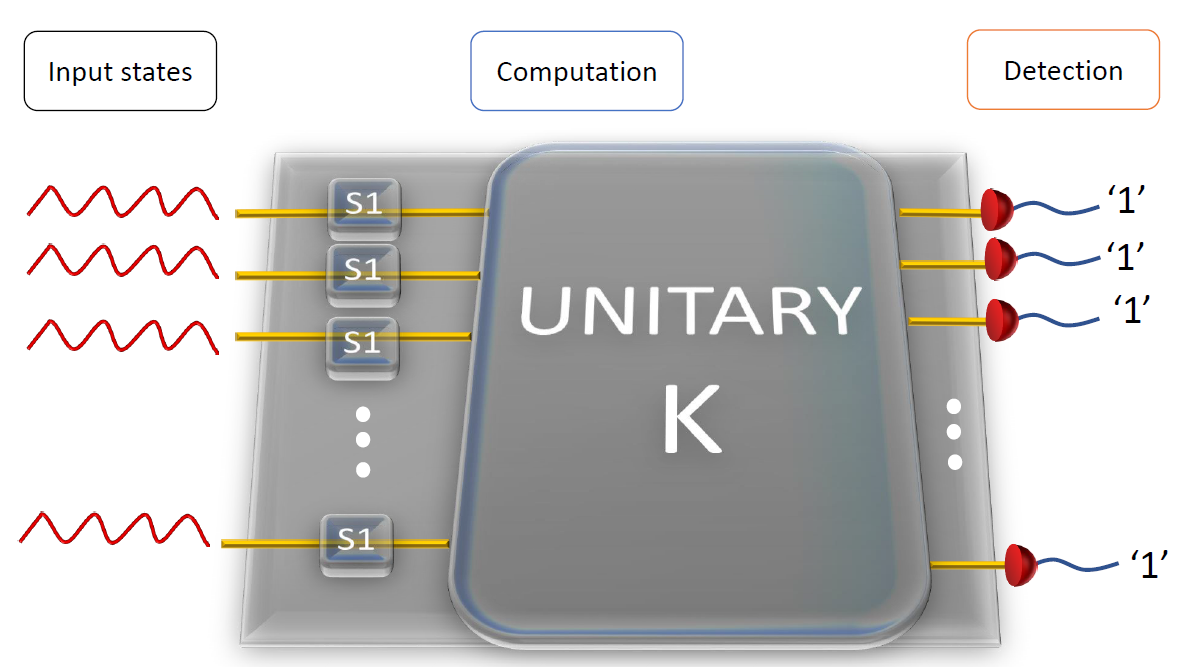}}
  \caption{ $2M$-mode canonical circuit that generates pure Gaussian covariance matrices $\s_{c A^{\oplus 2}}$. $2M$ single-mode squeezed states are injected   into a linear optics network $K$ to produce the target covariance matrix. The squeezing parameters $r_k$ ($r_k$ could be zero and corresponds to a vacuum input)   are completely determined by the eigenvalues of the adjacency matrix $A$ and  the rescaling parameter~$c$. In the output, $2M$ photon number detectors count the photon   number and evaluate the probability of measuring a single photon in every output mode.   }
  \label{fig:PureCircuit}
\end{figure*}

Generally, an adjacency matrix does not satisfy the three conditions discussed in Sec. \ref{sec:SampleMixState} and thus does not map to a mixed Gaussian covariance matrix. The strategy to overcome this difficulty is to embed $A$ into a bigger adjacency matrix which can be mapped to a valid covariance matrix. While there are many ways to do that, one straightforward way is to take two copies of the desired adjacency matrix $A$ and consider a new adjacency matrix $A^{\oplus 2} = A \oplus A$ instead. In this case, the new adjacency matrix satisfies the three conditions in Sec. \ref{sec:SampleMixState} and more importantly it maps to a pure Gaussian covariance matrix $\s_{c A^{\oplus 2}}$ (see Appendix \ref{sec:GraphToCM}). 
A pure Gaussian covariance matrix is easier to prepare than a mixed one. The price we have to pay is to double the number of modes, which means the circuit would be twice as large.

As a special case of the decomposition shown in Fig. \ref{fig:MixCircuit}, a pure Gaussian covariance matrix can be produced by injecting single-mode squeezed states (or vacuum) into a linear optics network. The squeezing parameter $r_k$ of the $k$-th input single-mode squeezed state is determined by the eigenvalues of the adjacency matrix $A$ and the rescaling parameter $c$ as $e^{2r_k} = (1+c|\lambda_k|)/(1-c|\lambda_k|)$, where $\lambda_k$ are the eigenvalues of the adjacency matrix $A$. Therefore given an adjacency matrix we can immediately determine the input squeezed states and the linear optics network according to the Euler decomposition, as shown in Fig. \ref{fig:PureCircuit}. We then detect the output state and evaluate the probability of measuring a single photon in each output mode. The hafnian of the adjacency matrix $A$ is related to the probability via
\begin{equation}\label{eq:probAoplA}
\mathrm{Pr}_{cA^{\oplus2}}(\underbrace{1,\dots, 1}_{2M}) = \prod_{k=1}^{2M} \sqrt{1-c^2 \lambda_k^2} ~c^{2M} \haf^{2}{A}\,.
\end{equation}

\subsection{Scalability}

The magnitude of the sampling probability determines the hardness of performing the experiment, e.g., a lower probability will take a longer time to collect the data.
We find that the sampling probability depends on the eigenvalues of the adjacency matrix, as well as the rescaling parameter $c$ (see Appendix \ref{app:SGCM}). In general, $c$ is smaller than one so the probability exponentially decreases as the size of the adjacency matrix 
increases. To have an optimal probability one has to increase the value of $c$. However, $c$ is restricted by $0 < c < 1/\la_1$, so the best one can do is to make $c$ as close as possible to $1/\la_1$. The value of $c$ is critical in terms of how tractable the problem becomes due to the exponential scaling. Such an effect is not unexpected. We probably cannot hope for an efficient quantum algorithm to count the number of perfect matchings for all graphs as a function of the graph size and connectedness. However, the way it manifests itself is intriguing.

The largest graph eigenvalue appears to be determining the hardness of sampling the corresponding covariance matrix (see the conjecture in Appendix \ref{app:scalability}). Note that spectral graph theory~\cite{brouwer2011spectra} is a rich field and the largest eigenvalue of the graph's adjacency matrix can play an important role~\cite{cvetkovic1990largest}.
Based on these insights, we can further `hack' the adjacency matrix by i) decreasing the maximal eigenvalue such that the matrix's hafnian remains the same, or ii) back-calculating the hafnian from multiple equivalent sampling probabilities. A more exhaustive list of possible improvements will be subject of future work.

The first strategy is to add/subtract a diagonal matrix to the adjacency matrix $A$, e.g., $A \rightarrow A + d \, \bbI$, which does not change the hafnian since self-loops are irrelevant
for the number of perfect matchings. This modification may result in a negative matrix and, correspondingly, we need to modify the constraint of c to be $0 < c < 1/|\lambda_1|$
(see Appendix \ref{app:scalability} for details).

The second strategy is to extend the adjacency matrix $A$ of size $2M$ and create a bigger adjacency matrix $\mathcal{A}$ of size $2nM$, where $n$ is an integer (see Appendix \ref{app:SubgraphSampling} for details).
$\mathcal{A}$ is chosen such that it has many submatrices that are equal to the adjacency matrix $A$. Since each of these submatrices is related to the hafnian of the adjacency
matrix $A$, we can sum up all the probabilities corresponding to sampling one of these submatrices. The amplification of the total probability is achieved. The price we have to pay is an increase in the number of modes $n$ times.

\subsection{Example: complete graphs}

To illustrate the method we have developed, consider a concrete example: estimating the number of perfect matchings of a complete graph by sampling the probability
of photon number distribution. Note that the number of perfect matchings in complete graphs is given by an analytical expression.
However, in this section we consider the case of complete graphs to illustrate our methods for scalabilty. Furthermore, in Appendix \ref{app:SubgraphSampling-2}, we give an explicit example of a noncomplete graph.
For a complete graph with $2M$ vertices, denoted as $K_{2M}$, all vertices are connected to each other and all entries of its adjacency matrix
$A_{K_{2M}}$ are one except the diagonal elements.
The number of perfect matchings is known to be $(2M-1)!!$. We estimate
the magnitude of the probability and the hardness of sampling.

According to the conditions discussed in Sec. \ref{sec:SampleMixState}, the adjacency matrix of a complete graph maps to a mixed covariance matrix. It is then
straightforward to calculate the probability by using Eq. (\ref{eq:probA}) and the number of perfect matchings. The rescaling parameter $c$ can be chosen to maximize the
probability. We can also modify the adjacency matrix by subtracting a diagonal matrix to reduce the absolute value of the maximal eigenvalue. The results for a 20 vertex
complete graph are shown in Fig. \ref{fig:K20}. Modifying the adjacency matrix significantly increases the probability.

\begin{figure}[ht!]
\includegraphics[width=8.6cm]{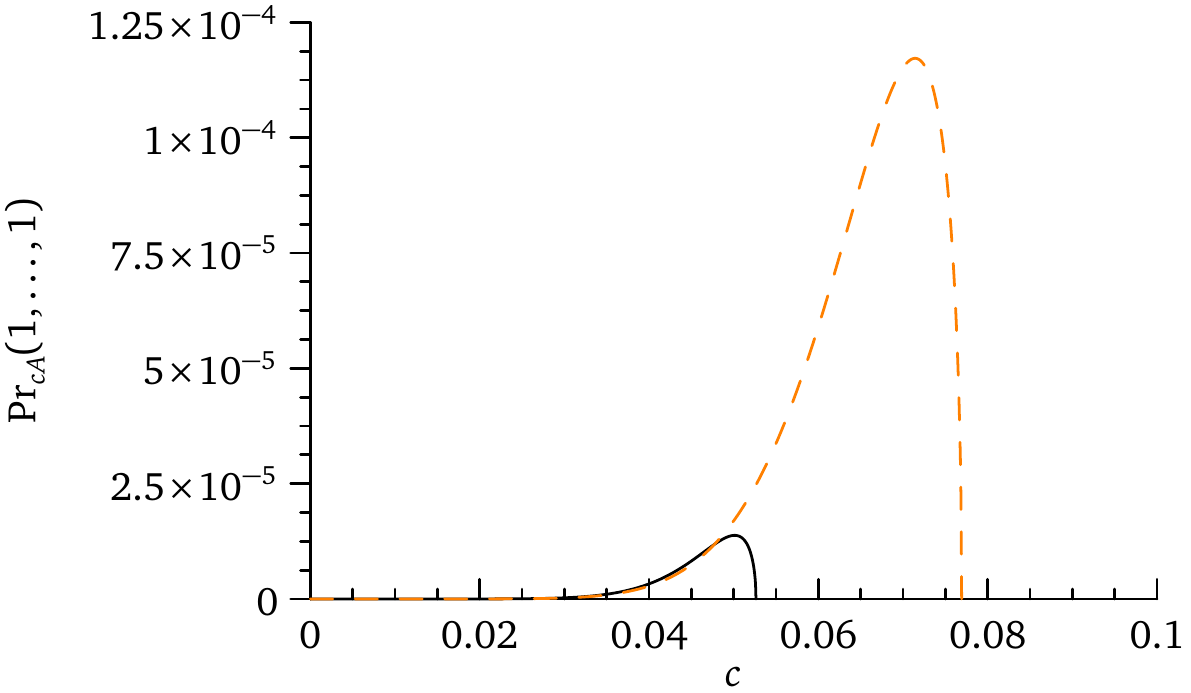}
\caption{ The probability of photon measurement for the complete graph on twenty vertices as a function of the rescaling parameter $c$, defined by $A_{K_{20}}  \to c A_{K_{20}} $, where $A_{K_{20}} $ is the graph's adjacency matrix. The black curve corresponds to original adjacency matrix and the dashed orange curve corresponds to modifying  the adjacency matrix as $A_{K_{20}} \rightarrow A_{K_{20}}- 6\, \bbI_{20}$. }
\label{fig:K20}
\end{figure}

We can also take two copies of the adjacency matrix $A_{K_{2M}}$ and sample a $2M$-mode pure Gaussian covariance matrix. However, we found that this is not optimal
for complete graphs. In fact, we can extend $K_{2M}$ to a bigger complete graph $K_{2n M}$. We take two copies of $A_{K_{2 n M}}$,
subtract a diagonal matrix to $A_{K_{2 n M}} \oplus A_{K_{2 n M}}$ and then sample a $2 n M $-mode pure covariance matrix. To relate to the hafnian of
$K_{2M}$, we have to consider the probability of measuring single photons in $2M$ modes and zero photons in the other modes. The total probability increases because we have $\binom{2nM}{2M}$ ways of selecting $2M$ modes from $2nM$ modes. The price we have to pay is to increase the size of the circuit. 

We find that the maximal squeezing that is required to achieve maximal measurement probability as a function of the size of a complete graph is given by
$10 \log_{10}\bigg[\frac{\sqrt{(2m+1)(n+1)}+\sqrt{2m(n+1)-1}}{\sqrt{(2m+1)(n+1)}-\sqrt{2m(n+1)-1}} \bigg]$ dB. 
Fig. \ref{fig:SqueezingMaxPro} shows an example for $n=16$. It can be seen that even for a relatively small-size graph, the amount of squeezing is demanding. From an experimental perspective, it is interesting to know how much probability one can achieve for a limited amount of squeezing. Fig. \ref{fig:ProSqueezing} shows the measurement probability as a function of the squeezing for a given complete graph. It can be seen that for a moderate amount of squeezing, e.g., 10 dB, the probability can go up to about $10^{-5}$ for a complete graph with $40$ vertices, and therefore
can be sampled in a very short time. Currently, one has access to gigahertz detectors, which means sampling at a rate of $10^4$ per second, where the longer the sampling
time the higher the accuracy one can achieve.

\begin{figure}[ht!]
\includegraphics[width=8.6cm]{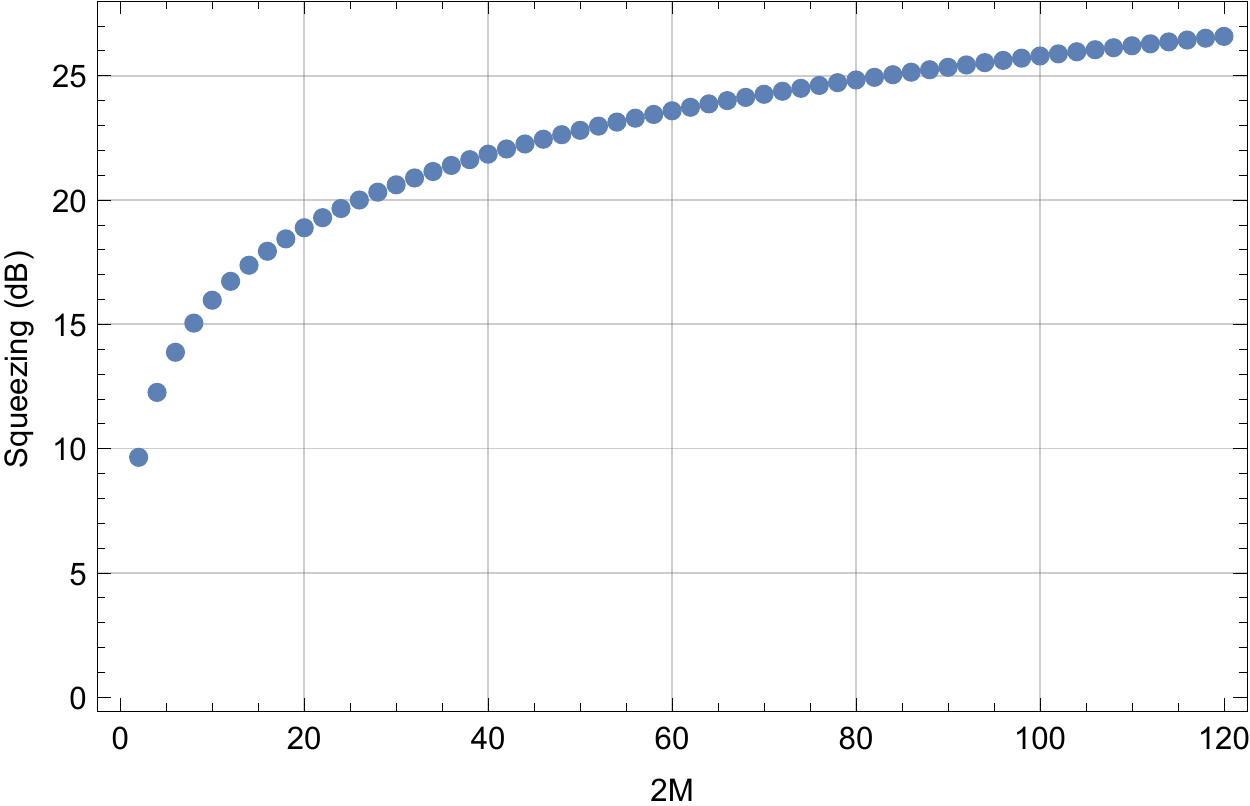}
\caption{ The maximum squeezing that is required to achieve the maximal measurement probability as a function of graph size $2M$. We assume that a complete graph $K_{2M}$ is embedded in a bigger
complete graph $K_{2nM}$ with $n=16$. Optimal rescaling factors are chosen to obtain the maximum probability. }
\label{fig:SqueezingMaxPro}
\end{figure}

\begin{figure}[ht!]
\includegraphics[width=8.6cm]{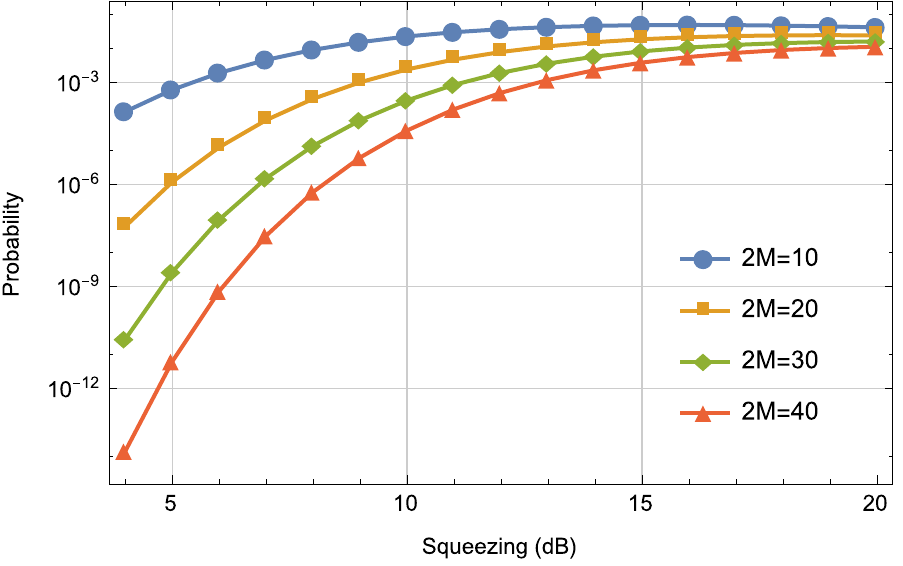}
\caption{ The maximum measurement probability as a function of the amount of squeezing for different graph sizes $2M$. Here, we assume that a complete graph $K_{2M}$ is embedded in a bigger
complete graph $K_{2nM}$ with $n=16$. Optimal rescaling factors are chosen to obtain the maximum probability. }
\label{fig:ProSqueezing}
\end{figure}

Surprisingly, we find that the discussed strategies can significantly increase the sampling probability for an extremely large complete graph (for a complete graph with one edge removed as well).
Fig.~\ref{fig:K5000} shows an example for a complete graph on 5000 vertices.
The maximal probability can go up to about $4\times 10^{-5}$. The number of modes required to achieve this high probability is $320000$  ($n=64$). 
Such a large circuit is challenging for current technologies -- the maximal squeezing that is required to achieve the maximal
probability is about 42.94 dB. Without this procedure the sampling of a well connected graph with more than one hundred vertices would be hopelessly unlikely.

\begin{figure}[ht!]
\includegraphics[width=8.6cm]{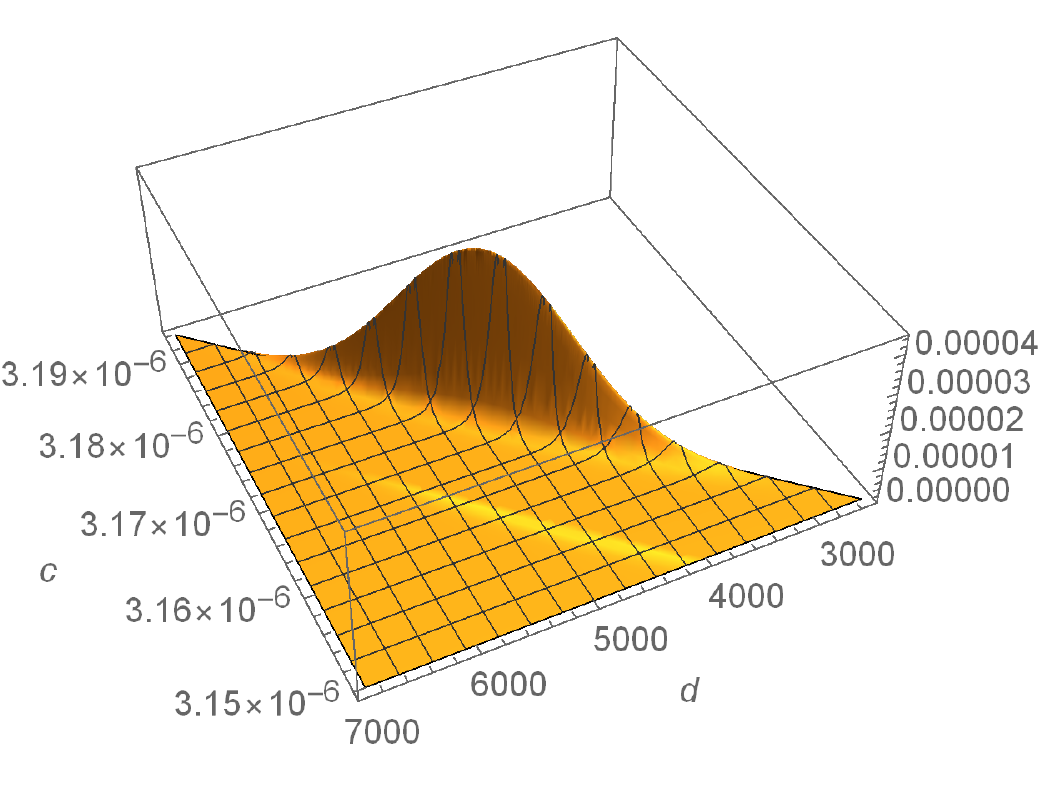}
\caption{ The probability of photon measurement for the complete graph of 5000 vertices. Here, $c$ is the rescaling parameter for the adjacency matrix and $d$ is the scaling factor for subtracting the identity matrix from the adjacency matrix (see also Appendix \ref{app:SubgraphSampling}). These rescalings correspond to the same graph in terms of perfect matchings but make the probability of detection
more favorable. The vertical axis shows the probability. }
\label{fig:K5000}
\end{figure}

\section{Conclusion}\label{sec:conclusion}

In this work, we have shown how to use a quantum device of squeezed states, optical interferometer, and photon counting measurements to estimate an important, and classically hard to compute, property of arbitrary graphs.
Based on the relation between Gaussian Boson Sampling and graph theory, we developed  a systematic and programmable method to map an arbitrary adjacency matrix of an undirected graph
to a continuous-variable Gaussian state. The number of perfect matchings of the graph can be directly estimated by sampling the photon number distribution of the
Gaussian state. The resources (input states and circuit) required to produce the desired Gaussian state are completely determined by the adjacency matrix and its
variants. 
Although the sampling probability decreases exponentially as the size of the graph increases, we have shown strategies to boost it to reasonably high values
by a suitable rescaling of the adjacency matrix and by linearly increasing the size of the circuit.

Estimating the number of perfect matchings of a large general graph is difficult for a classical computer. As a consequence of the problem of counting the perfect matchings being in the complexity class \#P, we expect the time to reach an exact solution to be exponential in the size of the graph on any physical device. 
We note that the computational hardness of Boson Sampling was conjectured via an embedding of Gaussian random matrices into larger random unitary ensembles \cite{aaronson2011computational}.
However, as hafnians are harder to compute than permanents, we conjecture that there exist graphs that correspond to important applications such as logistics, social networks, or chemistry and that show a quantum speedup. 
While exponential speedups could be possible for certain graphs, we expect that a special purpose quantum device offers a direct route to estimate the number of perfect matchings with better practical performance than a general purpose classical computer. 
Thus, the method we have developed may be used to achieve such a task faster and with lower energy consumption using a relatively small photonic device instead of a massive supercomputer.
Hence, our method has the potential to achieve practical quantum supremacy \cite{harrow2017quantum} for a problem with a wide range of applications.

\section*{Acknowledgements}

We thank Juan Miguel Arrazola and Nicol\'as Quesada for fruitful discussions.

\appendix
\vspace{0.5cm}

\section{Map an adjacency matrix to a covariance matrix}\label{sec:GraphToCM}

The hafnian  of a matrix is defined as following:
\begin{defi}\label{def:haf}
  Let $A=(a_{ij})$ be a $2M\times2M$ matrix and let $\varsigma$ denote a partition of the set $(1,\dots,2M)$ into unordered disjoint pairs.  There are $(2M-1)!!$ different such partitions. Then
  \begin{equation}\label{eq:shaf}
  \haf{A}=\sum_\varsigma\prod_{i=1}^{M}a_{\varsigma{(2i-1)}\varsigma{(2i)}}
  \end{equation}
  is the short hafnian.
\end{defi}
\begin{rem}
  The nomenclature comes from~\cite{schultz1992topological}. 
  The definition coincide for symmetric $0$\,-$1$ matrices and the prescription from~\cite{hamilton2016gaussian} yields the short hafnian.
\end{rem}
Following~\cite{hamilton2016gaussian}, the hafnian  can be obtained by sampling the output probability distribution of photons coming out from the multiport whose input is a Gaussian squeezed state. The relation between the covariance matrix $\s_A$ and the sampling matrix $A$ is given by Eq. (\ref{eq:sigmaA-1}). Here the covariance matrix is defined in the operator basis $\boldsymbol{\xi} = (a_1, \dots, a_M, a_1^{\dag}, \dots, a_M^{\dag})$,
\begin{equation}
\sigma_{A}(i,j) \equiv \frac{1}{2} \big\langle \big\{ \xi_i, \xi_j^{\dag} \big\} \big\rangle,
\end{equation}
where $\{\}$ stands for the anticommutator and we assume no displacements. The basis $\boldsymbol{\xi}$ will be called the Heisenberg basis in this paper.
Another operator basis used here will be called quadrature basis defined as $\tilde{\boldsymbol{\xi}}=(x_1, p_1,\dots, x_M, p_M)$.  The two bases are related by~\cite{weedbrook2012gaussian}
\begin{subequations}
\begin{align}
      a_j & ={1\over\sqrt{2}}(x_j+ip_j), \\
      a_j^{\dag} & ={1\over\sqrt{2}}(x_j-ip_j).
\end{align}
\end{subequations}
In the quadrature basis, the covariance matrix $\tilde{\s}_A$ is defined as
\begin{equation}
\tilde{\sigma}_{A}(i,j) \equiv \frac{1}{2} \big\langle \big\{ \tilde{\xi}_i, \tilde{\xi}_j^{\,\dag} \big\} \big\rangle.
\end{equation}

A symplectic operation in both the quadrature and Heisenberg basis preserves the symplectic form \cite{weedbrook2012gaussian}
\begin{equation}\label{eq:OmegaXPXP}
    \Omega=\bigoplus_{j=1}^{M}\om_j\equiv\bigoplus_{j=1}^{M}\begin{bmatrix}
                        0 & 1 \\
                        -1 &  0\\
                      \end{bmatrix}.
\end{equation}

Our particular interest is in developing a procedure to calculate the hafnian of an arbitrary adjacency matrix $A$ by mapping it to a valid covariance matrix and then sampling the photon number distribution. However, unless by a chance, a general $A$ does not map to a physical Gaussian covariance matrix by~\eqref{eq:sigmaA} since a covariance matrix in the Heisenberg basis has to be Hermitian (it is real if the adjacency matrix is real), positive-definite and satisfies the uncertainty principle \cite{simon1994quantum}, namely,
\begin{equation}\label{eq:UncertaintyPrinciple}
\sigma_A + \frac{1}{2} Z_{2M} \ge 0,
\end{equation}
where $Z_{2M}$ is defined in the Heisenberg basis as
\begin{equation}
Z_{2M} =
\begin{bmatrix}
  \bbI_M & 0 \\
  0 & - \bbI_M
\end{bmatrix}.
\end{equation}

We show in the following what conditions $A$ has to be satisfied to
guarantee a valid covariance matrix.
\begin{lem}\label{lem:XAsymmetric}
   Let the block form of a $2M\times 2M$ real symmetric matrix $A$  be
   \begin{equation}\label{eq:Ablockdiag}
     A=\begin{bmatrix}
         A_{11} & A_{12} \\
         A_{21} & A_{22} \\
       \end{bmatrix}.
   \end{equation}
   Then $\bbI_{2M}-X_{2M}A$ is symmetric if the square submatrices $A_{ij}$ of dimension $M$ satisfy $A_{12}=A_{21}$ and $A_{22}=A_{11}^\top$.
\end{lem}
   \begin{proof}
     By imposing symmetry on
     $$
     X_{2M}A=\begin{bmatrix}
         A_{21} & A_{22} \\
         A_{11} & A_{12} \\
       \end{bmatrix}
     $$
     we immediately find the submatrix equalities. The difference of two symmetric matrices is symmetric.
   \end{proof}
\begin{cor}\label{cor:XAcommute}
  Let $A$ be symmetric and satisfy the conditions of Lemma \ref{lem:XAsymmetric}. 
  Then $[X_{2M},A]=0$.
\end{cor}
\begin{proof}
  As $X_{2M}A$ is symmetric as in Lemma \ref{lem:XAsymmetric}, note that if a product of two symmetric matrices is symmetric then they commute.
\end{proof}

\begin{thm}\label{thm:XAcond}
    Let $A$ be a $2M\times2M$ real symmetric matrix and assume that $X_{2M}A$ is also symmetric. Then $\s_A>0$ iff
    the eigenvalues $\la_k$ of $A$ satisfy $|\la_k| < 1, \forall k$.
\end{thm}
\begin{rem}
  One can show that $(\bbI_{2M}-X_{2M}A)^{-1}$ is symmetric if $\bbI_{2M}-X_{2M}A$ is, and therefore $\sigma_A$ is symmetric (and Hermitian) from \eqref{eq:sigmaA-1}.
  In Lemma~\ref{lem:XAsymmetric} we show which $A$'s it holds for, but it is still a rather special condition on an adjacency matrix $A$. In the following
  we will show how to overcome these limitations.
\end{rem}
\begin{proof}
  By writing the second identity matrix in \eqref{eq:sigmaA-1} as
  \begin{align*}
  \bbI_{2M} &= (\bbI_{2M}-X_{2M}A)(\bbI_{2M}-X_{2M}A)^{-1} \nonumber\\
  &=
  (\bbI_{2M}-X_{2M}A)^{-1} (\bbI_{2M}-X_{2M}A),
  \end{align*}
  we can express $\s_A$ as
    \begin{align}\label{eq:sigmaAdifferent}
      \sigma_A &= {1\over2}(\bbI_{2M}+X_{2M}A)(\bbI_{2M}-X_{2M}A)^{-1} \nonumber\\
      &=
      {1\over2}(\bbI_{2M}-X_{2M}A)^{-1}(\bbI_{2M}+X_{2M}A).
    \end{align}
    It is known~\cite{horn2012matrix} that two matrices are simultaneously diagonalizable iff they commute. From Corollary \ref{cor:XAcommute}, $X_{2M}$ and $A$ commute, therefore they can be diagonized simultaneously by an orthogonal matrix and the eigenvalues of $X_{2M}A$, labeled as $\la^{X}_{k}$,  are the products of eigenvalues of  $X_{2M}$ and $A$. Since the eigenvalues of $X_{2M}$ are either $+1$ or $-1$, the eigenvalues of $X_{2M}A$ satisfy $|\lambda^{X}_{k}| = |\lambda_k|$.

From \eqref{eq:sigmaAdifferent}, the orthogonal matrix that diagonalizes  $X_{2M}A$ also diagonalizes $\sigma_A$, giving eigenvalues
    \begin{equation}
    \lambda^{\sigma}_{k} = \frac{1}{2} \bigg( \frac{1+\lambda^{X}_{k}}{1-\lambda^{X}_{k}} \bigg).
    \end{equation}
    Therefore $\lambda^{\sigma}_{k}>0$ iff $|\lambda^{X}_{k}|<1$, namely, $\sigma_A>0$ iff $|\lambda_{k}|<1, \forall k$.
\end{proof}
\begin{defi}
   A matrix $A=(a_{ij})$ is nonnegative if $\forall i,j\  a_{ij}\geq0$.
\end{defi}
\begin{thm}[Perron-Frobenius,~\cite{horn2012matrix}]\label{thm:PerFrob}
  Let $A$ be nonnegative and irreducible and let $\la_1$  be the largest eigenvalue of $A$. Then $\la_1$ is positive and $\la_1\geq|\la_j|,\forall j$.
\end{thm}
\begin{prop}\label{prop:Aspectrum}
  Let submatrices $A_{11}$ and $A_{12}$ in~\eqref{eq:Ablockdiag} be symmetric and let $\la_1$  be the largest eigenvalue of $A$.
  Then $\sigma_{cA} > 0$ for $0<c<1/\la_1$.
\end{prop}
\begin{proof}
From Theorem \ref{thm:PerFrob} and the assumption $0<c<1/\la_1$, the maximum eigenvalue of $cA$ is positive and smaller than one, and eigenvalues $\la'_k = c|\la_k|$ of $cA$ satisfy $0<\la'_k < 1, \forall k$. According to
Theorem \ref{thm:XAcond}, we have $\sigma_{cA} >0$.
\end{proof}
Let $A$ be the adjacency matrix of an undirected graph~\cite{west2001introduction}. Then, in general, we do not get a valid covariance matrix from~Eq.~\eqref{eq:sigmaA-1}. If $X_{2M}A$ is symmetric we can use Proposition~\ref{prop:Aspectrum} and rescale $A\mapsto cA$ such that $\s_{cA}>0$. Although a Hermitian and positive definite $\s_{cA}$ does not guarantee a valid covariance matrix, we will see in the following that $\s_{cA}$ does represent a covariance matrix for some of the adjacency matrices $A$.

\begin{lem}\label{lem:validCM}
  Let $A_{11}$ and $A_{12}$ commute, then $\s_{cA}$ is a valid covariance matrix if $A_{12}\ge 0$.
\end{lem}
\begin{proof}
  If $A_{11}$ and $A_{12}$ commute, there exists an orthogonal matrix $S$ that diagonalizes $A_{11}$ and $A_{12}$ simultaneously,
    \begin{subequations}
    \begin{align}
       S A_{11} S^\top &= \diag{[f_1, f_2, \dots, f_M]}, \\
       S A_{12} S^\top &= \diag{[ h_1, h_2, \dots, h_M]},
    \end{align}
    \end{subequations}
  where $f_k$ and $h_k$ are the eigenvalues of $A_{11}$ and $A_{12}$, respectively. Then we find
  \begin{equation}
  \left( S\oplus S \right) \,A \,\big( S^\top \oplus S^\top \big) =
      \begingroup
    \renewcommand{\arraystretch}{1.2}
    \begin{bmatrix}
      \bigoplus_{k=1}^{M} f_k &  \bigoplus_{k=1}^{M} h_k \\
      \bigoplus_{k=1}^{M} h_k &  \bigoplus_{k=1}^{M} f_k
    \end{bmatrix}
    \endgroup
  \end{equation}
  and it is obvious that $f_k \pm h_k$ are eigenvalues of $A$.
  From \eqref{eq:sigmaA}, the transformation of $\s_{cA}$ is
  \begin{align}
  &\left( S\oplus S \right) \,\s_{cA} \,\left( S^\top \oplus S^\top  \right) \nonumber\\
  &=
  \begingroup
    \renewcommand{\arraystretch}{1.5}
    \begin{bmatrix}
      \frac{1}{2} \bigoplus_{k=1}^{M} \frac{1-c^2 h_k^2 + c^2 f_k^2}{(1-c h_k)^2-c^2f_k^2} &  \bigoplus_{k=1}^{M} \frac{c f_k}{(1-c h_k)^2-c^2f_k^2} \\
      \bigoplus_{k=1}^{M} \frac{c f_k}{(1-c h_k)^2-c^2 f_k^2} &  \frac{1}{2} \bigoplus_{k=1}^{M} \frac{1-c^2 h_k^2 + c^2f_k^2}{(1-c h_k)^2-c^2 f_k^2}
    \end{bmatrix}.
  \endgroup \nonumber\\
  \end{align}
  The eigenvalues of $\s_{cA}$ can be calculated as
  \begin{align}\label{eq:SpecA}
  \spec{\s_{cA}} &=\frac{1}{2}  \bigg[ \frac{1+c( h_k + f_k) }{1- c( h_k + f_k)}, \,  \frac{1+c( h_k - f_k) }{1- c( h_k - f_k)} \bigg] \nonumber\\
  \end{align}
  for $\ k = 1, 2, \dots, M$.
  Note that $Z_{2M}$ is invariant,
  \begin{equation*}
    \left( S\oplus S \right) \, Z_{2M} \,\big( S^\top  \oplus S^\top  \big) = Z_{2M}.
  \end{equation*}
  Therefore we have
   \begin{align}
     &\left( S\oplus S \right) \, \big(Z_{2M} \s_{cA}\big) \,\big( S^\top  \oplus S^\top  \big) \nonumber\\
     &=Z_{2M} \left( S\oplus S \right) \,\s_{cA} \,\big( S^\top  \oplus S^\top  \big) \nonumber \\
     &= \begin{bmatrix}
      \frac{1}{2} \bigoplus_{k=1}^{M} \frac{1-c^2 h_k^2 + c^2 f_k^2}{(1-c h_k)^2-c^2f_k^2} &  \bigoplus_{k=1}^{M} \frac{c f_k}{(1-c h_k)^2-c^2f_k^2} \\
      -\bigoplus_{k=1}^{M} \frac{c f_k}{(1-c h_k)^2-c^2 f_k^2} &  -\frac{1}{2} \bigoplus_{k=1}^{M} \frac{1-c^2 h_k^2 + c^2f_k^2}{(1-c h_k)^2-c^2 f_k^2}
    \end{bmatrix}. \nonumber\\
  \end{align}
  We then can find the symplectic eigenvalues of $\s_{cA}$ (the eigenvalues of $|Z_{2M} \s_{cA}|$) as
  \begin{equation}
  \nu_k = \frac{1}{2} \sqrt{\frac{(1+c h_k)^2 - c^2 f_k^2}{(1-c h_k)^2 - c^2 f_k^2}}.
  \end{equation}
  Hence  $h_k \geq 0$  in order to satisfy $\nu_k \ge 1/2$.
\end{proof}
\begin{rem}
  The covariance matrix is pure only when all eigenvalues of $A_{12}$ are zero, which means $A_{12}=0$ and $A$ is in a block diagonal form. In the case when some eigenvalues of $A_{12}$ are positive the Gaussian state is mixed.
\end{rem}

When $\s_{cA}$ is a valid covariance matrix, we recover $\haf{A}$ from $\haf{[cA]}$ via Lemma~\ref{lem:scaledHaf}.

\begin{lem}\label{lem:scaledHaf}
  Let $c\in\bbR$. Then  $\haf{[cA]}=c^M\haf{A}$ where $\dim{A}=2M$.
\end{lem}
\begin{proof}
  Follows directly from~Eq.~\eqref{eq:shaf}.
\end{proof}
The analysis can be extended to nonnegative matrices. At first, this may be a surprising move as the adjacency matrix of an undirected graph has no negative entries. But as long as the hafnian of such a matrix coincides with the hafnian of the desired adjacency matrix (or can be use to find it in an efficient way) and the corresponding covariance matrix is more advantageous (that will be the case) we can consider it. We will discuss such modifications in Appendix \ref{app:scalability} but in general we will assume that the requirements on the modified $A$ in Lemma~\ref{lem:XAsymmetric}, Corollary~\ref{cor:XAcommute} and Theorem~\ref{thm:XAcond} are satisfied. But since we allowed $A$ to be negative, its eigenvalues may become negative and Proposition~\ref{prop:Aspectrum} must be strengthened. This is because it cannot rely on Perron-Frobenius theorem. But the modification is just cosmetic. As before, we get $0<c<1/\la_1$ in Proposition~\ref{prop:Aspectrum} but now $\la_1$ is the maximal absolute eigenvalue of $X_{2M}A$, which on the other hand, equals to the maximal absolute value of the eigenvalues of $A$. This is because the multiplication by $X_{2M}$ at most changes the signs of some of the eigenvalues of $A$.

Even if the covariance matrix is positive definite and Hermitian/symmetric, it may not be a valid Gaussian covariance matrix. Every covariance matrix has to satisfy the uncertainty inequality~\cite{simon1994quantum} which translates into the symplectic eigenvalues being greater than one half, see Eq.~\eqref{eq:UncertaintyPrinciple}. Remarkably, we can take care of this issue by doubling the dimension of the adjacency matrix and the corresponding covariance matrix turns out to be a pure state. As we will see, this procedure is a panacea of a sort: so far we considered only cases where $X_{2M}A$ is symmetric. If it is not the case, and this can routinely happen for legitimate adjacency matrices or its modification in the sense of the previous paragraph, the doubling prescription makes the covariance matrix pure as well. Therefore it works for all adjacency matrices. If $X_{2M}$ and $A$ do not commute then again $0<c<1/\la_1$ holds and $\la_1$ is the maximal absolute eigenvalue of $X_{2M}A$ but it may not be directly related to the eigenvalues of $A$.
\begin{lem}\label{lem:hafProduct}
  Let $(V_1,E_1)$ and $(V_2,E_2)$ denote graphs with adjacency matrices $A_1,A_2$. Then  $\haf{[A_1\oplus A_2]}=\haf{A_1}\haf{A_2}$.
\end{lem}
\begin{proof}
  The adjacency matrix of $(V_1,E_1)\cup(V_2,E_2)$ is $A_1\oplus A_2$. Count the number of perfect matchings for both of them. As they are independent, the total number of perfect matchings of $A$ is their product.
\end{proof}
Hence, the proposal is to take two copies of the desired adjacency matrix~$A$ and consider~$A^{\oplus2}=A\oplus A$ instead. We not only get the needed block-diagonal form~\cite{hamilton2016gaussian} but also $A^{\oplus2}_{12}=A^{\oplus2}_{21}$ and $A^{\oplus2}_{22}=(A^{\oplus2}_{11})^\top$ are satisfied. By further using Proposition \ref{prop:Aspectrum} and multiplying by $0<c<1/\la_1$, we get a valid covariance matrix $\s_{cA^{\oplus2}}$. Since $\dim{\s_{cA^{\oplus2}}}=4M$, it is a state given by the canonical $2M$-mode circuit ($2M$ single-mode squeezers followed by a multiport), thus merely doubling the number of modes for any $M$.

\begin{lem}\label{lem:pureAoplusA}
 Let $A$ be an adjacency matrix of a graph. The covariance matrix $\sigma_{cA^{\oplus2}}$ constructed from $c(A\oplus A)$ represents a pure Gaussian state if $c$ is positive and smaller than the maximal eigenvalue of $A$.
\end{lem}
\begin{proof}
 The conclusion directly follows from the proof of Lemma \ref{lem:validCM}.
\end{proof}

Since the state is pure, $\s_{cA^{\oplus2}}$ can be decomposed as
\begin{equation}
\sigma_{cA^{\oplus2}} = R \s_\mathrm{in} R^\top.
\end{equation}
Here $\s_\mathrm{in}$ written in the quadrature basis represents an input pure covariance matrix of $M$ independent single-mode squeezed states:
\begin{widetext}
\begin{equation}\label{eq:sigma}
\s_\mathrm{in} = \diag{}{\bigg[ \frac{1}{2} \bigg( \frac{1+c|\lambda_1| }{1-c |\lambda_1|}\bigg),
\frac{1}{2} \bigg( \frac{1-c|\lambda_1| }{1+c |\lambda_1|}\bigg),
\hdots,
\frac{1}{2} \bigg( \frac{1+c|\lambda_{2M}| }{1-c |\lambda_{2M}|}\bigg),
\frac{1}{2} \bigg( \frac{1-c|\lambda_{2M}| }{1+c |\lambda_{2M}|}\bigg) \bigg]}.
\end{equation}
\end{widetext}
The amount of squeezing in each mode can be determined by
\begin{equation}
e^{2 r_k} = \frac{1+c|\lambda_k| }{1-c |\lambda_k|},
\end{equation}
where $r_k$ is the squeezing parameter in the $k$-th mode. Note that if $\lambda_k = 0$,  then $r_k = 0$, corresponding to a vacuum state. When
$c\rightarrow1/\la_1$, the maximum squeezing parameter $r_1\to\infty$, corresponding to an infinitely squeezed state. $R$ is an orthogonal symplectic matrix and can be implemented as a multimode passive interferometer. Therefore, a pure Gaussian state with covariance matrix $\s_{cA^{\oplus2}}$ can be generated by an array of single-mode squeezed states or vacuum states followed by a sequence of beam-splitters and phases shifts.

\section{Sampling the Gaussian covariance matrix}\label{app:SGCM}

\subsection{Pure Gaussian covariance matrices}

Our main goal in this paper is to estimate the hafnian of an arbitrary adjacency matrix by sampling the photon number distribution of a certain Gaussian state. The procedure to calculate the hafnian of an adjacency matrix is: (a) find a corresponding covariance matrix of $c (A \oplus A)$,
(b) construct an optical circuit to generate the Gaussian state, (c) sample the output photons and find the probability of the pattern with single photon in each output mode, (d) derive the hafnian via the relation
\begin{align}\label{eq:probAoplA}
\Pr_{cA^{\oplus2}}\nolimits(\underbrace{1, \dots, 1}_{2M})
&=  \prod_{k=1}^{2M} \sqrt{1-c^2 \lambda_k^2} ~\haf{[c(A\oplus A)]} \nonumber\\
&= \prod_{k=1}^{2M} \sqrt{1-c^2 \lambda_k^2} ~c^{2M}\haf^2{A},
\end{align}
namely, the hafnian is 
\begin{eqnarray}
\haf{A} = c^{-M} \sqrt{\Pr_{cA^{\oplus2}}\nolimits(1, \dots, 1)} \prod_{k=1}^{2M} (1-c^2 \lambda_k^2)^{-1/4}, \nonumber\\
\end{eqnarray}
where we have used the fact that $\bar n! = 1$ and $\det \s_Q$ can be easily calculated using the eigenvalues of the covariance matrix.

\subsection{Mixed Gaussian covariance matrices}

According to Lemma \ref{lem:validCM}, there exists a nontrivial class of adjacency matrices $A$ whose covariance matrix $\s_{cA}$ reresents a legitimate Gaussian mixed state. From \eqref{eq:SpecA}, it is easy to find the eigenvalues of $\s_Q = \s_{cA} + \bbI_{2M}/2$,
\begin{equation}
\spec{\s_Q} =  \bigg[ \frac{1}{1-c(f_k + h_k)}, \, \frac{1}{1+c(f_k - h_k)} \bigg],
\end{equation}
for $k = 1, 2, \dots, M$.
Therefore the determinant of $\s_Q$ is
\begin{equation}
\det{\s_Q} = \prod_{k=1}^{M} \frac{1}{\big[1-c(f_k + h_k)\big] \big[ 1+c(f_k - h_k)\big]}.
\end{equation}
By using Eq. \eqref{eq:ProHafOrig}, the probability of measuring one photon in each output mode is
\begin{align}\label{eq:probA-app}
&\Pr_{cA}\nolimits(\underbrace{1,\dots,1}_{M}) \nonumber\\
&= \prod_{k=1}^{M} \sqrt{[1-c(f_k + h_k)] [ 1+c(f_k - h_k)]} ~c^M \haf{} A.
\end{align}

\section{Scalability hacks}\label{app:scalability}

The largest graph eigenvalue seems to be determining the hardness of sampling the corresponding covariance matrix. We are thus led to the following hypothesis:
\begin{conj}\label{CojHardness}
  The largest eigenvalue of an adjacency matrix~$A$ determines the hardness (time complexity) of estimating the hafnian of~$A$ encoded in an optical circuit.
\end{conj}

Based on these insights, we can further `hack' the adjacency matrix to decrease the maximal eigenvalue such that its hafnian remains the same or can be efficiently back-calculated but the sampling probability can be obtained more efficiently. For example, if we construct
\begin{equation}\label{eq:Aprime}
A'=A+\diag{[d_1,\dots,d_{2M}]}
\end{equation}
such that $d_i=d_{i+M}$ then for any $A$ and $d_i\in\bbR$ we get $\haf{}{A}=\haf{}{A'}$. This follows from the irrelevance of the self-loops for the number of perfect matchings or directly from Definition~\ref{def:haf}. 
Another possibility is based on the following claim.
\begin{lem}\label{lem:MultiplyNumber}
  Let $A=(a_{ij})$ be a symmetric matrix of an even dimension and let $B=(b_{ij})$ have the same entries and size as $A$ except for the $k$-th row and column where $\forall i,j$, $b_{kj}=na_{kj}$ and $b_{ik}=na_{ik}$. Then
  \begin{equation}
    \haf{B}=n\,\haf{A}.
  \end{equation}
\end{lem}
\begin{proof}
  From the definition of hafnian, Eq.~\eqref{eq:shaf}, we see that the $k$-th row appears in each product exactly once. The same holds for the $k$-th column.
\end{proof}
\begin{rem}
  Note that the way of calculating the hafnian of $A$ from~\cite{hamilton2016gaussian} we use here is slightly different from~\eqref{eq:shaf}. Instead, the partial derivatives result in a product of $(a_{i\varsigma(i)}+a_{\varsigma(i)i})/2$ which becomes~\eqref{eq:shaf} for a symmetric $A$. This is why we required both the $k$-th row and column to be multiplied by~$n$ in~Lemma~\ref{lem:MultiplyNumber}.
\end{rem}

\subsection{Four vertices complete graph}
    We illustrate how to obtain the hafnian from the measurement statistics for $K_4$ -- the complete graph  on four vertices. Its adjacency matrix reads
    \begin{equation}\label{eq:AforK4}
    A=\begin{bmatrix}
        0 & 1 & 1 & 1 \\
        1 & 0 & 1 & 1 \\
        1 & 1 & 0 & 1 \\
        1 & 1 & 1 & 0 \\
      \end{bmatrix}.
    \end{equation}
    Since $M=2$, we get from~\eqref{eq:probAoplA}
    \begin{equation}\label{eq:Prob4MOrig}
    \mathrm{Pr}_{cA^{\oplus2}} (1, 1, 1, 1) =  (1-9c^2)^{1/2} ~(1-c^2)^{3/2}c^{4} \haf^2{A}.
    \end{equation}
    Here $c$ is a free parameter, so we can choose it to maximize the measuring probability. Once we know the probability $\mathrm{Pr}_{cA^{\oplus2}} (1, 1, 1, 1)$ by sampling  the Gaussian covariance matrix $\s_{cA^{\oplus2}}$, we can obtain the hafnian of the four vertices complete graph $K_4$ via Eq. \eqref{eq:Prob4MOrig}.

Here we would like to consider a slightly different problem: how to optimize the measuring probability for a given graph.
For a complete graph $K_4$ we know the hafnian is $\haf{A}=3$. By substituting the hafnian into \eqref{eq:Prob4MOrig}, we find the probability is a function of $c$, which we plot as a black line in Fig.~\ref{fig:Prob4M}. We can see that there exists a particular $c$ between $0$ and $1/3$ that maximizes the probability, and the probability goes to zero if $c\to0$ or $c\to1/3$. The $c \to 0$ limit corresponds to vacuum input, the probability evidently goes to zero. The $c \rightarrow 1/3$ limit corresponds to infinite squeezing in one of the input modes, therefore the probability of measuring low particle number also tends to be zero.

We can add/subtract an arbitrary diagonal matrix to the adjacency matrix to decrease the maximal eigenvalue without changing the hafnian. As a result of a numerical optimization considering all the constraints on a Gaussian covariance matrix, we get the highest probability of measuring the hafnian for
    \begin{equation}\label{eq:ApforK4}
    A' =
     \begin{bmatrix}
        -2/3 & 1 & 1 & 1 \\
        1 & -2/3 & 1 & 1 \\
        1 & 1 & -2/3 & 1 \\
        1 & 1 & 1 & -2/3 \\
      \end{bmatrix}.
    \end{equation}
    Given $\spec{A'}=(7/3,-5/3,-5/3,-5/3)$, then $\s_{c A^{\prime \oplus2}}$ is a valid Gaussian covariance matrix provided $0<c<3/7$. The probability of measuring one photon in each output mode becomes
    \begin{equation}\label{eq:Prob4MSubtract}
    \mathrm{Pr}_{cA^{\prime \oplus2}}(1, 1, 1, 1) = {1\over81}(49c^2-9)^{1/2}(25c^2-9)^{3/2}c^4\haf^2{A}.
    \end{equation}
    The probability as a function of $c$ is plotted as an orange dashed line in Fig.~\ref{fig:Prob4M}. Similarly, we find that there exists a $c$ between $0$ and $3/7$ that maximizes the probability. Indeed, the maximum probability is greater than the maximum probability corresponding to the original adjacency matrix~$A$. This illustrates the power of adding a diagonal matrix to the original $A$.
    \begin{figure}[ht!]
    \includegraphics[width=8.3cm]{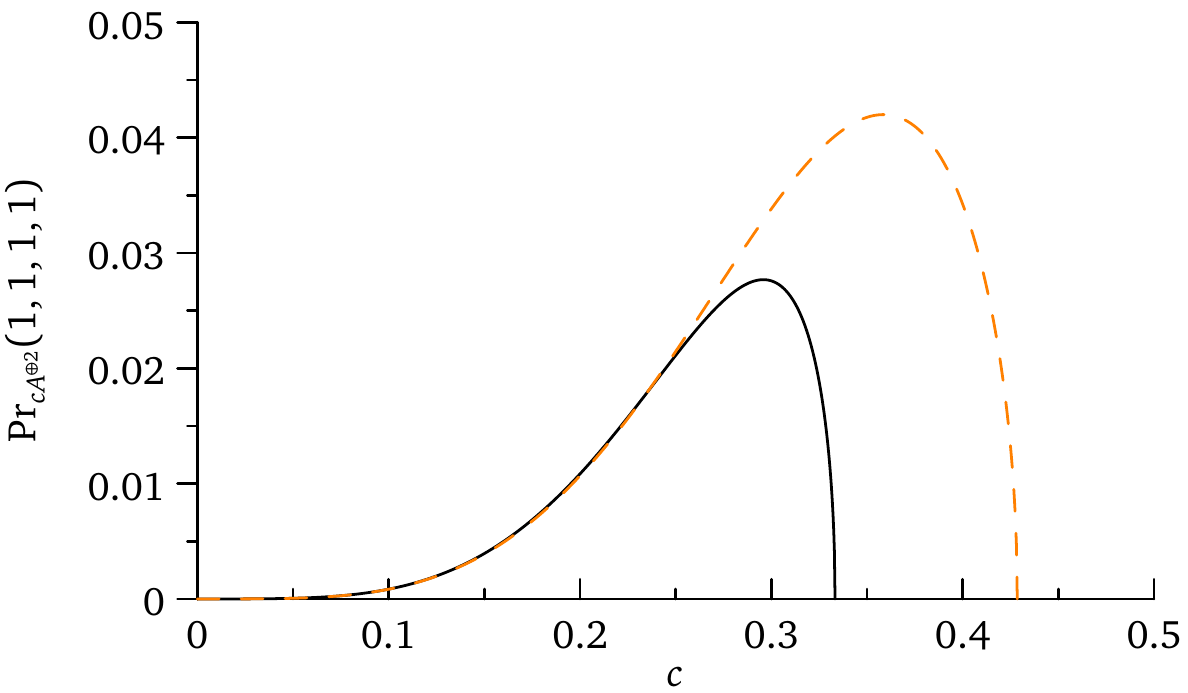}
    \caption{Sampling a pure state: probability of measurement of one photon in each output mode as a function of the rescaling parameter $c$. The black line corresponds to
    the original adjacency matrix $A$ and the orange dashed line corresponds to $A'$ in~\eqref{eq:ApforK4}. }
    \label{fig:Prob4M}
    \end{figure}

\begin{figure}[h]
    \includegraphics[width=8.3cm]{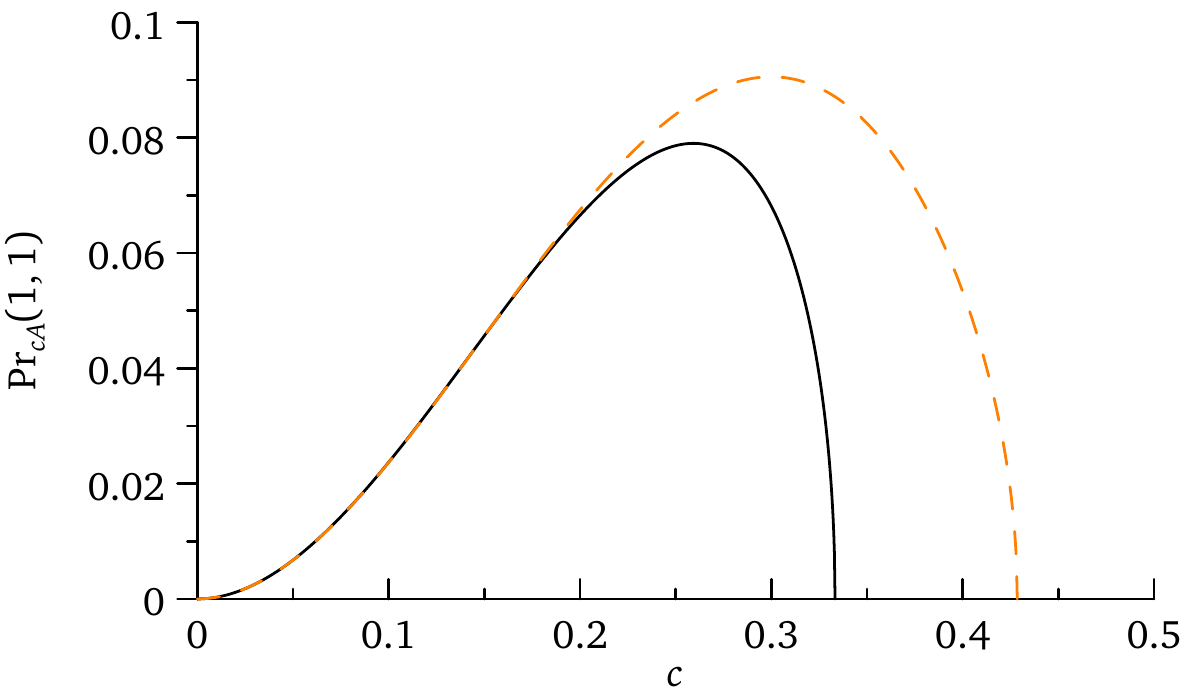}
    \caption{Sampling a mixed state: probability of measuring one photon in each output mode without doubling the modes. The black line corresponds to the original adjacency matrix $A$, Eq.~\eqref{eq:AforK4}, while the orange dashed line corresponds to~\eqref{eq:ApforK4}. }
    \label{fig:Prob2Mmixed}
    \end{figure}

    To see how advantageous a sampling of a mixed state is, we use Eq.~\eqref{eq:probA-app} and find
    \begin{equation}
    \mathrm{Pr}_{cA}(1, 1)=(1-c)\sqrt{(1-3c)(1+c)}\,c^2\haf{A}.
    \end{equation}
    By sampling the optimized covariance matrix $\s_{cA'}$ for $A'$ again given by~\eqref{eq:ApforK4} we find
    \begin{equation}
    \mathrm{Pr}_{cA'}(1, 1)={1\over9}(3-5c)\sqrt{(3-7c)(3+5c)}\,c^2\haf{A}.
    \end{equation}
    Fig.~\ref{fig:Prob2Mmixed} compares sampling probability of these two methods. Compared to a pure state sampling method in Fig.~\ref{fig:Prob4M}, the probability of the desired detection event is about $30 \%$ higher.
        We have not found any advantage by using Lemma~\ref{lem:MultiplyNumber}.

\subsection{Twenty vertices complete graph}
  Let $K_{20}$ be the complete graph on twenty vertices. It is known that $\#PM(K_{20})=654729075$.  Highly connected graphs on many vertices are demanding for a classical computer if the number of perfect matchings is counted by brute force. 
  For a graph of size $20$ it is still possible to count the perfect matchings in seconds \cite{bjorklund2012counting}. How does our quantum-based approach perform? The adjacency matrix gives rise to a mixed ten-mode Gaussian covariance matrix ($M=10$). From Eq.~\eqref{eq:probA-app}, we find
  \begin{subequations}
  \begin{align}%
    \mathrm{Pr}_{cA}(1,\dots,1) & = (1-c)^5 (1+c)^{9/2}\sqrt{1-19 c}\,c^{10}\haf A_{K_{20}},\label{eq:AforK10}\\
    \mathrm{Pr}_{cA'}(1,\dots,1) & = (1-7c)^{5}(1+7c)^{9/2} \sqrt{1-13c}\,c^{10} \haf A_{K_{20}} \label{eq:AprimeforK10}.
  \end{align}
    \end{subequations}
  The first expression is based on the adjacency matrix for $K_{20}$ and the second one is $A'$ whose diagonal elements were optimized. The single-photon probabilities are depicted in Fig.~\ref{fig:K20} and we see a dramatic increase in the probability of detection of the desired pattern. Hence, with a reasonably fast quantum sources and detectors we could sample the circuit and gather the necessary statistics within seconds.

\section{A proof-of-principle experiment}\label{app:PPE}

In this appendix, we study how to construct a circuit and what input states are required to generate the covariance matrix corresponding to a given adjacency matrix.
To illustrate the essential steps without involving complicated technical details, we study a concrete example of the complete graph on four vertices.
We find the eigenvalues of $A$ from~\eqref{eq:AforK4} to be $\spec{A}=(3,-1,-1,-1)$ and so $\lambda_1=3$.
The covariance matrix in the Heisenberg basis is found to be
\begin{widetext}
\begin{equation}\label{eq:CMforK4}
    \s_{cA}=\begin{bmatrix}
    \begin{array}{cccc}
     -\frac{3 c^3+c^2+c-1}{2 (c-1) \left(3 c^2+2 c-1\right)} & \frac{c}{-3 c^2-2 c+1} & \frac{2 c^2}{3 c^3-c^2-3 c+1} & \frac{c}{-3 c^2-2 c+1} \\
     \frac{c}{-3 c^2-2 c+1} & -\frac{3 c^3+c^2+c-1}{2 (c-1) \left(3 c^2+2 c-1\right)} & \frac{c}{-3 c^2-2 c+1} & \frac{2 c^2}{3 c^3-c^2-3 c+1} \\
     \frac{2 c^2}{3 c^3-c^2-3 c+1} & \frac{c}{-3 c^2-2 c+1} & -\frac{3 c^3+c^2+c-1}{2 (c-1) \left(3 c^2+2 c-1\right)} & \frac{c}{-3 c^2-2 c+1} \\
     \frac{c}{-3 c^2-2 c+1} & \frac{2 c^2}{3 c^3-c^2-3 c+1} & \frac{c}{-3 c^2-2 c+1} & -\frac{3 c^3+c^2+c-1}{2 (c-1) \left(3 c^2+2 c-1\right)} \\
    \end{array}
    \end{bmatrix}.
\end{equation}
\end{widetext}
The corresponding symplectic eigenvalues of are
\begin{equation}\label{eq:sympleigsCMforK4}
  \nu_{1,2}={1\over2},\quad\nu_{3,4}=\frac{1}{2} \sqrt{\bigg( \frac{1+3c}{1-3c} \bigg)\bigg( \frac{1+c}{1-c} \bigg)}.
\end{equation}
This implies the Gaussian state represented by $\sigma_{cA}$ is a mixed state~\cite{weedbrook2012gaussian}. To obtain a pure state, we follow~Lemma~\ref{lem:pureAoplusA} and construct a covariance matrix $\s_{cA^{\oplus 2}}$ for $c(A\oplus A)$ in the Heisenberg basis using Eq.~\eqref{eq:sigmaA-1}:
\begin{equation}\label{eq:CMexample}
\s_{cA^{\oplus 2}} = \frac{1}{(1-c^2)(1-9c^2)}
\begin{bmatrix}
  C & D \\
  D & C
\end{bmatrix}
\end{equation}
where $C$ and $D$ are $4\times 4$ matrices:
\begin{widetext}
\begin{equation}
C=
\begin{bmatrix}
  \frac{1-4c^2-9c^4}{2} & 2c^2 & 2c^2 & 2c^2 \\
  2c^2 & \frac{1-4c^2-9c^4}{2} & 2c^2 & 2c^2 \\
  2c^2 & 2c^2 & \frac{1-4c^2-9c^4}{2} & 2c^2 \\
  2c^2 & 2c^2 & 2c^2 & \frac{1-4c^2-9c^4}{2}
\end{bmatrix},
~~~~~~~~
D=
\begin{bmatrix}
  6c^3 & c(1-3c^2) & c(1-3c^2) & c(1-3c^2) \\
  c(1-3c^2) & 6c^3 & c(1-3c^2) & c(1-3c^2) \\
  c(1-3c^2) & c(1-3c^2) & 6c^3 & c(1-3c^2) \\
  c(1-3c^2) & c(1-3c^2) & c(1-3c^2) & 6c^3
\end{bmatrix}.
\end{equation}
The symplectic eigenvalues of $\s_{cA^{\oplus 2}}$ are $\nu_k=1/2$ for all $k$, showing that the state is pure. By ordinary diagonalization we obtain

\begin{equation}\label{eq:specK4double}
K \s_{cA^{\oplus2}} K^\top
= \diag{\bigg[\frac{1+3c}{2(1-3c)},\frac{1-3c}{2(3c+1)},\frac{1+c}{2(1-c)},\frac{1-c}{2(1+c)},\frac{1+c}{2(1-c)},\frac{1-c}{2(1+c)},\frac{1+c}{2(1-c)},\frac{1-c}{2(1+c)}\bigg]}.
\end{equation}
\end{widetext}

We are going to find the corresponding canonical circuit generating the covariance matrix $\sigma_{cA^{\oplus 2}}$. Eigenspectrum~\eqref{eq:specK4double} is an array squeezing parameters for four single-mode squeezers and the parameter $c$ tunes how many photons are injected to the multiport $K$ that follows. By virtue of diagonalization, the eigenspectrum is in the quadrature basis. But the properties of the multiport, namely how to be able to decompose it to elementary beam splitters and phase shifts, are best seen in the Heisenberg basis. By transforming it we get, as expected~\cite{hamilton2016gaussian}, $K=T\oplus\overline{T}$, where
\begin{equation}\label{eq:multipK4double}
  T=\begin{bmatrix}
     \frac{\frac{1}{2}-\frac{i}{2}}{\sqrt{2}} & \frac{\frac{1}{2}-\frac{i}{2}}{\sqrt{2}} & \frac{\frac{1}{2}-\frac{i}{2}}{\sqrt{2}} & \frac{\frac{1}{2}-\frac{i}{2}}{\sqrt{2}} \\
     -\frac{1}{2}+\frac{i}{2} & 0 & 0 & \frac{1}{2}-\frac{i}{2} \\
     -\frac{\frac{1}{2}-\frac{i}{2}}{\sqrt{3}} & 0 & \frac{1-i}{\sqrt{3}} & -\frac{\frac{1}{2}-\frac{i}{2}}{\sqrt{3}} \\
     -\frac{\frac{1}{2}-\frac{i}{2}}{\sqrt{6}} & \left(\frac{1}{2}-\frac{i}{2}\right) \sqrt{\frac{3}{2}} & -\frac{\frac{1}{2}-\frac{i}{2}}{\sqrt{6}} & -\frac{\frac{1}{2}-\frac{i}{2}}{\sqrt{6}}
    \end{bmatrix}.
\end{equation}
By decomposing $T$ into an array of six beam splitters and some phase shifts according to~\cite{clements2016optimal}, we find
\begin{equation}\label{eq:Tdecomposed}
  T=B_{23}(-\theta_6)B_{34}(-\theta_5)\Delta B_{12}(\theta_4)B_{23}(\theta_3)B_{34}(\theta_2)B_{12}(\theta_1),
\end{equation}
where
\begin{align}
&(\t_i) = \\
&\bigg({-2 \tan^{-1}{\frac{1}{3}}}, -\frac{\pi }{2}, -2 \tan^{-1}{\sqrt{5}}, -2\tan^{-1}{\sqrt{\frac{3}{2}}},-\pi ,-\frac{\pi}{3} \bigg)\nonumber
\end{align}
and $\D=e^{-i\pi/4}\diag{[1,1,-1,1]}$. The phase operator $\D$ can be further moved left or right if the sign of some $\t_i$'s are swapped.

The outlined procedure of doubling the dimension of~$A$ (or, equivalently, doubling the number of modes of the circuit) will work for any~$A$ as a result of~Lemma~\ref{lem:pureAoplusA}. Besides the obvious downside of a slightly bigger circuit, the main disadvantage is an increase in complexity of finding the hafnian of~$A$ as discussed in the previous section. Hence, if $\s_{cA}$ represents a valid Gaussian covariance matrix like in the case of~\eqref{eq:AforK4} for $0<c<1/3$, it is highly advantageous to prepare the corresponding circuit instead of~\eqref{eq:specK4double} and~\eqref{eq:multipK4double}. Since $A$ is not block-diagonal, $\s_A$ is not canonical and we need the following results to be able to derive the circuit generating it.
\begin{thm}[Williamson,~\cite{williamson1936algebraic}]\label{thm:Williamson}
  Let $\s$ be a positive-definite real matrix of dimension $2M$. Then there exists a symplectic matrix $S$  such that
  \begin{equation}\label{eq:Williamson}
    S^\top\s S=\bigoplus_{i=1}^{M}\nu_i\bbI_2
  \end{equation}
\end{thm}
The RHS of Eq.~\eqref{eq:Williamson} is inevitably in the quadrature basis. A special case of the real Schur decomposition is stated as follows~\cite{youla1961normal,zumino1962normal}:
\begin{thm}\label{thm:SchrReal}
  For every real antisymmetric matrix $M$ there exists an orthogonal matrix $Q$ such that
  \begin{equation}\label{eq:Schur}
    QMQ^\top=\tilde{M},
  \end{equation}
  where
  $$
  \tilde{M}=\bigoplus_{i=1}\begin{bmatrix}
                       0 & \kappa_i \\
                       -\kappa_i & 0 \\
                     \end{bmatrix},
  $$
  is called a normal antisymmetric form and $\kappa_i\in\bbR$.
\end{thm}

Once we find the diagonalizing symplectic matrix we invert~\eqref{eq:Williamson} by writing
\begin{equation}\label{eq:WilliamsonInv}
  \s=S^{-\top}\s_\mathrm{in}S^{-1},
\end{equation}
where $\s_\mathrm{in}$ is the input thermal state. The circuit is completed by the symplectic singular value decomposition (SVD) of $S$~\cite{arvind1995real,braunstein2005squeezing,xu2003svd}. In more detail, let $Sp(2M,\bbR)$ denote the group of (real) symplectic matrices. As a result of the SVD, a symplectic matrix $S\in Sp(2M,\bbR)$ can be expressed as $S=K\Sigma L^\top$, where $\Sigma\geq0$ and $K,L\in K(2M)$ are elements of the orthogonal symplectic group~$K(2M)=SO(2M)\cap Sp(2M,\bbR)$~\cite{arvind1995real,xu2003svd}. The SVD is not unique and it may happen that the matrices $K,L$  are orthogonal but not symplectic. To obtain it in the desirable form, we rewrite the SVD of $S$ as the left polar decomposition: $S=(K\Sigma K^\top)(KL^\top)= PO$, where $P=(SS^\top)^{1/2}$ is real positive-definite, symmetric and $O$ is orthogonal. Moreover, $P$ and $O$ are symplectic~\cite{ferraro2005gaussian} and $\Sigma$ is a diagonal matrix containing the eigenvalues of $P$. Hence, the desired circuit to generate $\sigma$ from $\s_\mathrm{in}$ is $K\Sigma^{-1}L^\top$, where $K,L^\top$ are multiports and $\Sigma^{-1}$ is an array of single-mode squeezers.

Note, however, that $\s_A$ is in the Heisenberg basis. By transforming it to the quadrature basis and multiplying by 2 we define
  \begin{equation}\label{eq:CMforK4Quadrature}
    \s'=2\begin{bmatrix}
    \begin{array}{cccc}
     \frac{3 c^2+1}{-6 c^2-4 c+2} & 0 & \frac{2 c}{-3 c^2-2 c+1} & 0 \\
     0 & \frac{c+1}{2-2 c} & 0 & 0 \\
     \frac{2 c}{-3 c^2-2 c+1} & 0 & \frac{3 c^2+1}{-6 c^2-4 c+2} & 0 \\
     0 & 0 & 0 & \frac{c+1}{2-2 c} \\
    \end{array}
    \end{bmatrix}.
  \end{equation}
  To find $S$ from Eq.~\eqref{eq:Williamson}, using the notation of~Eq.~\eqref{eq:WilliamsonInv}, we follow the standard procedure (see e.g.~\cite{simon1999congruences}) and write
  \begin{equation}\label{eq:S}
    S=(\s'/2)^{-1/2}Q\s_\mathrm{in}^{1/2},
  \end{equation}
  where $Q$ is an orthogonal matrix needed to transform the antisymmetric matrix $\s^{\prime-1/2}\Omega\s^{\prime-1/2}$ into a normal antisymmetric form according to Theorem~\ref{thm:SchrReal}. Given $S$, we apply the SVD and recover the symplectic circuit in the quadrature basis. We find
  \begin{equation}\label{eq:Q}
    Q=\begin{bmatrix}
    \begin{array}{cccc}
     0 & -\frac{1}{\sqrt{2}} & 0 & \frac{1}{\sqrt{2}} \\
     \frac{1}{\sqrt{2}} & 0 & -\frac{1}{\sqrt{2}} & 0 \\
     0 & \frac{1}{\sqrt{2}} & 0 & \frac{1}{\sqrt{2}} \\
     -\frac{1}{\sqrt{2}} & 0 & -\frac{1}{\sqrt{2}} & 0 \\
    \end{array}
      \end{bmatrix}
  \end{equation}
  which is indeed orthogonal and
  \begin{widetext}
  \begin{equation}\label{eq:SexaK4}
    S=\frac{1}{\sqrt{2}} \begin{bmatrix}
    \begin{array}{cccc}
     0 & -\sqrt{\frac{1+c}{1- c}} & 0 & \sqrt[4]{\frac{(1+c)(1-3c)}{(1-c)(1+3 c)}} \\
     \sqrt{\frac{1-c}{c+1}} & 0 & -\sqrt[4]{\frac{(1-c)(1+3c)}{(1+c)(1-3 c)}} & 0 \\
     0 & \sqrt{\frac{c+1}{1- c}} & 0 & \sqrt[4]{\frac{(1+c)(1-3c)}{(1-c)(1+3 c)}} \\
     -\sqrt{\frac{1-c}{c+1}} & 0 & -\sqrt[4]{\frac{(1-c)(1+3c)}{(1+c)(1-3 c)}} & 0 \\
    \end{array}
    \end{bmatrix},
  \end{equation}
  where we check that $S$ is symplectic and diagonalizes $\s$ according to Theorem~\ref{thm:Williamson} leading to
  \begin{equation}\label{eq:symplDiagExpl}
     S^\top\s S=\diag{\bigg[\frac{1}{2},\frac{1}{2},\frac{1}{2} \sqrt{\bigg( \frac{1+3c}{1-3c} \bigg)\bigg( \frac{1+c}{1-c} \bigg)}, \frac{1}{2} \sqrt{\bigg( \frac{1+3c}{1-3c} \bigg)\bigg( \frac{1+c}{1-c} \bigg)}\bigg]}
  \end{equation}
  written in the quadrature basis in agreement with~\eqref{eq:sympleigsCMforK4}. The SVD of $S$ then yields
    \begin{equation}\label{eq:SVDK4a}
    K=\begin{bmatrix}
    \begin{array}{cccc}
     0 & \frac{1}{\sqrt{2}} & 0 & -\frac{1}{\sqrt{2}} \\
     -\frac{1}{\sqrt{2}} & 0 & \frac{1}{\sqrt{2}} & 0 \\
     0 & \frac{1}{\sqrt{2}} & 0 & \frac{1}{\sqrt{2}} \\
     -\frac{1}{\sqrt{2}} & 0 & -\frac{1}{\sqrt{2}} & 0 \\
    \end{array}
    \end{bmatrix},\quad
    L^\top=\begin{bmatrix}
    \begin{array}{cccc}
     0 & 0 & 1 & 0 \\
     0 & 0 & 0 & 1 \\
     1 & 0 & 0 & 0 \\
     0 & 1 & 0 & 0 \\
    \end{array}
    \end{bmatrix}
  \end{equation}
  and
    \begin{equation}\label{eq:SVDK4b}
    \Sigma^{-1}=\diag{{\bigg[\bigg(\frac{(1+c)(1-3c)}{(1-c)(1+3 c)}\bigg)^{1/4}, \bigg(\frac{(1+c)(1-3c)}{(1-c)(1+3 c)}\bigg)^{-1/4}, \sqrt{\frac{1+c}{1-c}} ,\sqrt{\frac{1-c}{1+c}}\bigg]}}.
  \end{equation}
  \end{widetext}
  The matrices are still in the quadrature basis and $K$ and $L^\top$ are orthogonal and symplectic as desired. Then, the circuit can be read off of these matrices in the Heisenberg basis and we arrive at the simple circuit in Fig.~\ref{fig:K4mixcirc}.
  \begin{figure}[h]
  \resizebox{7cm}{!}{\includegraphics{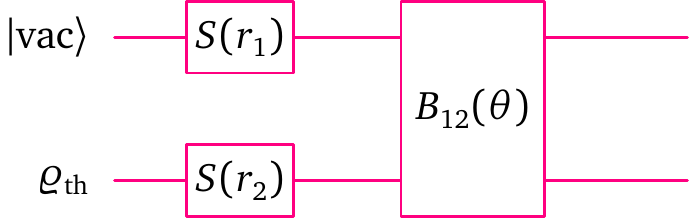}}
    \caption{The circuit to generate the covariance matrix in Eq.~\eqref{eq:CMforK4}. From~\eqref{eq:SVDK4b} we get $r_1=\ln{\sqrt{\frac{1+c}{1-c}}}$ and $r_2=\ln{\left(\frac{(1+c)(1-3c)}{(1-c)(1+3 c)}\right)^{1/4}}$ for the squeezing parameters. The thermal state density matrix is 
    $\varrho_{\text{th}} = \frac{2}{2\nu+1} \sum_{n=0}^{\infty} \big(\frac{2\nu-1}{2\nu+1} \big)^n |n\rangle \langle n|$ with
    $\nu= \frac{1}{2} \sqrt{\big( \frac{1+3c}{1-3c} \big)\big( \frac{1+c}{1-c} \big)}$. 
    From $K$ in~\eqref{eq:SVDK4a} we obtain $\theta=\pi/2$ and because $L^\top$ is a mere swap we took it into account by swapping the input modes.}\label{fig:K4mixcirc}
\end{figure}

\section{Subgraph sampling}\label{app:SubgraphSampling}
\subsection{Complete graphs}\label{app:SubgraphSampling-1}
\begin{defi}
  An $n$-extended complete graph $K^\uparrow_{2M}$ on $2M$ vertices is another, larger, complete graph $K_{2nM}$ such that $n\in\bbZ_{>0}$.
\end{defi}
Let $\euK_M=K_{2M}\oplus K_{2M}$ be two copies of a complete graph on $2M$ vertices. Then $\euK^\uparrow_M=\euK_{nM}$. The goal is to make the extension large enough to find a sufficient number of copies of the desired  graph. Following~Lemma~\ref{lem:scaledHaf} and~\ref{lem:hafProduct}, we find $\haf{[cA_{\euK_M}]}=c^{2M}((2M-1)!!)^2$ and then from~Eq.~\eqref{eq:ProHafOrig} we get
\begin{equation}\label{eq:ProbKext}
  \Pr{\euK^\uparrow_M}=\binom{2nM}{2M}{1\over \bar n! \sqrt{\det{\s_{Q,\euK^\uparrow_M}}}}c^{2M}((2M-1)!!)^2.
\end{equation}
The binomial coefficient tells us how many measurement patterns corresponding to the studied graph $\euK_M$ there is in~$\euK^\uparrow_{M}$. In order to calculate the determinant we have to find the eigenvalues of the adjacency matrix $cA_{\euK^\uparrow_M}$. Recall that
$$
\haf{[cA_{\euK^\uparrow_M}]}=\haf\,{[c(A_{\euK^\uparrow_M}-d\,\bbI_{4nM})]}
$$
and we are looking for the eigenvalues of the more general matrix on the RHS. We find
\begin{widetext}
\begin{equation}\label{eq:specKext}
  \spec{[c(A_{\euK^\uparrow_M}-d\bbI_{4nM})]}=\big((c(1+d))_{1}^{2nM-1}),(-c(1+d))_{1}^{2nM-1},\pm(c(2nM-1-d))\big)
\end{equation}
where the symbol $(x)^n_1$ stands for $(\underbrace{x, \hdots, x}_n)$.
From Eq.~\eqref{eq:sigma} and $\s_{Q,\euK^\uparrow_M}=\s_{\euK^\uparrow_M}+\bbI_{4n M}/2$ we get
\begin{equation}\label{eq:detKext}
  \det{\s_{Q,\euK^\uparrow_M}}=\big((-1+c^2 (-(2nM-1)+d)^2) (-1+c+c d)^{2nM-1} (1+c+c d)^{2nM-1}\big)^{-1}.
\end{equation}
\end{widetext}

\subsection{Complete graphs with one edge removed}\label{app:SubgraphSampling-2}
\begin{defi}
  Let $M_m$ be a space of real matrices of size $m$. A mapping $\chi_n:M_m\mapsto M_{nm}$ is called $n$-inflation, defined for  any $A=(a_{ij})\in M_m$ as
   \begin{equation}\label{eq:inflation}
     \chi_n(a_{ij})=a_{ij}\bbJ_n,
   \end{equation}
   where the RHS is a minor of size $n\times n$ and $\bbJ_n$ is an all-ones square matrix of size $n$.
\end{defi}
The hafnian of a general matrix is difficult to calculate. But by carefully removing the edges of a complete graph, whose hafnian is known, we can keep track of all the perfect matchings that are lost.  The simplest case is one edge removed and by doing so we discard $(2M-3)!!$ perfect matchings.
We shall name such a graph $R_{2M}$.  Hence $\haf{R_{2M}}=(2M-1)!!-(2M-3)!!=(2M-3)!!(2M-2)$. Due to the symmetry of complete graphs, we do not need to track the particular edge we got rid of in $R_{2M}$. Let its adjacency matrix be
\begin{equation}\label{eq:RemovedEdges}
  A_{R_{2M}}=\begin{bmatrix}
               0_2 & \bbJ_{2,2M-2} \\
               \bbJ_{2M-2,2} & A_{K_{2M-2}} \\
             \end{bmatrix}.
\end{equation}
Then we define
\begin{equation}\label{eq:RextendedInflated}
  A_{R^\uparrow_{2M}}=\begin{bmatrix}
               \chi_n(0_2) & \chi_n(\bbJ_{2,2M-2}) \\
               \chi_n(\bbJ_{2M-2,2}) & A_{K^\uparrow_{n(2M-2)}} \\
             \end{bmatrix}
\end{equation}
which is another adjacency matrix of dimension $2nM$ and finally we introduce $\euR^\uparrow_M=R^\uparrow_{2M}\oplus R^\uparrow_{2M}$ from which we will sample submatrices corresponding to the investigated graphs $R_{2M}$. To this end, we find
\begin{align}\label{eq:ProbRext}
  \Pr{\euR^\uparrow_M}&=\binom{2n}{2}\binom{n(2M-1)}{2M-2}{1\over \bar n! \sqrt{\det{\s_{Q,\euR^\uparrow_M}}}} \nonumber\\
  &\times c^{2M}((2M-3)!!(2M-2))^2,
\end{align}
where the first binomial coefficient `picks up' the right number of zeros (two zeros for any $R_{2M}$) and the second coefficient counts the number of ways the rest of $R_{2M}$ is constructed. The remaining task is to calculate the determinant of  $A_{\euR^\uparrow_M}$ and the following Lemma will assist us.
 \begin{widetext}
\begin{lem}
   Let
   \begin{equation}
    W=\begin{bmatrix}
               0_k & \bbJ_{k,\ell} \\
               \bbJ_{\ell,k} & A_{K_\ell} \\
             \end{bmatrix}-d\,\bbI_{k+\ell}
    \equiv
     \begin{bmatrix}
       D & B \\
       B^\top & A \\
     \end{bmatrix}.
   \end{equation}
   Then
   \begin{equation}\label{eq:specW}
     \spec{W}=\big((-1-d)_1^{\ell-1},(-d)_1^{k-1},1/2(-1-2d+\ell\pm\sqrt{1-2\ell+\ell^2+4\ell k})\big).
   \end{equation}
\end{lem}
\begin{proof}
  We calculate the characteristic polynomial of $W$:
  $$
  \det{[W-\la\bbI_{k+\ell}]}=\det{[D-\la\bbI_{k}-B(A-\la\bbI_{\ell})^{-1}B^\top]}\det{[A-\la\bbI_{\ell}]}=0.
  $$
  The second determinant reads
  \begin{equation}\label{eq:2ndCharPoly}
    \det{[A-\la\bbI_{\ell}]}=(\ell-1-d-\la)(1+d+\la)^{\ell-1}.
  \end{equation}
  For the first determinant we find
  $$
  (A-\la\bbI_{\ell})^{-1}={1\over(-\ell+1+d+\la)(1+d+\la)}\bbJ_\ell+{\ell-2-d-\la\over(-\ell+1+d+\la)(1+d+\la)}\bbI_\ell
  $$
  and so
  \begin{equation}
    B(A-\la\bbI_{\ell})^{-1}B^\top=-{\ell\over-\ell+1+d+\la}\bbJ_k.
  \end{equation}
  Therefore
  \begin{equation}
    \det{[D-\la\bbI_{k}-B(A-\la\bbI_{\ell})^{-1}B^\top]}
    =\det{\Big[{\ell\over-\ell+1+d+\la}\Big({-(d+\la)(-\ell+1+d+\la)\over\ell}\bbI_k+\bbJ_k\Big)\Big]}.
  \end{equation}
  By using $\det{[\a X_k]}=\a^k\det{X}_k$ and $\det{[\b\bbI_k+\bbJ_k]}=\b^{k-1}(\b+k)$ we find
  \begin{align}\label{eq:1stCharPoly}
    &\det{[D-\la\bbI_{k}-B(A-\la\bbI_{\ell})^{-1}B^\top]}\nonumber \\
    &={(-d-\la)^{k-1}(-\ell+1+d+\la)^{k-1}\over\ell^{k-1}}{\ell^k\over(-\ell+1+d+\la)^k}\Big({-(d+\la)(-\ell+1+d+\la)\over\ell}+k\Big).
  \end{align}
  By combining it with \eqref{eq:2ndCharPoly}  we finally get
  \begin{equation}\label{eq:bothPolys}
    \det{[W-\la\bbI_{k+\ell}]}=(-d-\la)^{k-1}\Big({-(d+\la)(-\ell+1+d+\la)\over\ell}+k\Big)(1+d+\la)^{\ell-1}=0
  \end{equation}
  and Eq.~\eqref{eq:specW} follows.
\end{proof}
To find the determinant in~\eqref{eq:ProbRext} we deduce from~\eqref{eq:specW}
\begin{equation}\label{eq:speccW}
     \spec{[cW]}=\big((c(-1-d))_1^{\ell-1},(-cd)_1^{k-1},c/2(-1-2d+\ell\pm\sqrt{1-2\ell+\ell^2+4\ell k})\big).
\end{equation}
and from Eq.~\eqref{eq:sigma} and $\s_{Q,\euR^\uparrow_M}=\s_{\euR^\uparrow_M}+\bbI_{4nM}/2$ we finally find
\begin{align}\label{eq:detRext}
  &\det{\s_{Q,\euR^\uparrow_M}}={16\over(-1 + c^2 d^2)^{k-1}(-1+c^2(1+d)^2)^{\ell-1}}\nonumber \\
  &\times{1\over(-4+c^2(-1-2d+\ell+\sqrt{1-2\ell+\ell^2+4\ell k})^2)(-4-c^2(-1-2d+\ell-\sqrt{1-2\ell+\ell^2+4\ell k})^2)},
\end{align}
where $k=2n$ and $\ell=n(2M-2)$.
\end{widetext}

\bibliographystyle{unsrt}

\bibliography{haf}

\end{document}